\documentclass[11pt]{article}
\usepackage{microtype}
\usepackage{titling}
\usepackage[english]{babel}
\usepackage{blindtext}
\usepackage{tikz}
\usetikzlibrary{trees}
\usepackage{mathdots}
\usepackage{amsmath,amsthm,amsfonts,bm}
\usepackage{stmaryrd}
\usepackage{xfrac}
\usepackage{url}
\usepackage{multirow}
\usepackage{subcaption}
\usepackage{fancyhdr} 
\usepackage[margin=1in]{geometry}
\usepackage[cal=pxtx]{mathalfa}
\usepackage{wrapfig}
\usepackage{thmtools, thm-restate}
\usepackage[]{algorithm2e}
\usepackage{mathrsfs}
\usepackage{enumitem}
\usepackage{hyperref}
\usepackage{cleveref}

\newcommand*\samethanks[1][\value{footnote}]{\footnotemark[#1]}
\newcommand{\email}[2]{\protect\href{mailto:#2}{#1}}

\newtheorem{thm}{Theorem}
\newtheorem*{thm*}{Theorem}
\newtheorem{lem}{Lemma}
\newtheorem{cor}{Corollary}
\newtheorem{prop}{Proposition}
\theoremstyle{definition}
\newtheorem{dfn}{Definition}

\newcommand{\reals}{{\mathbb{R}}}
\newcommand{\eps}{{\varepsilon}}
\newcommand{\E}{\mathbb{E}}

\newcommand{\gpi}{\Gamma_{\bm{\sigma}}}
\newcommand{\evirt}{E_{\text{virt}}}
\newcommand{\ephys}{E_{\text{phys}}}

\newcommand{\bigotilde}{\tilde{\mathcal{O}}}
\newcommand{\bigo}{\mathcal{O}}

\newcommand{\zmodn}{\mathbb{Z}/\left( N \right)}

  \newcommand{\pths}{\mathcal{P}}
  \newcommand{\pathcat}{\circ}
  
  \newcommand{\vlbleak}{{VLB}_{\text{leak}}}
  \newcommand{\unvlbleak}{{\overline{VLB}}_{\text{leak}}}

\newcommand{\Z}{\mathbb{Z}}

\newcommand{\C}{\mathbb{C}}

\newcommand{\ignore}[1]{}

\newcommand{\braces}[1]{\ensuremath{\left\{ #1 \right\}}}

\newcommand{\set}[1]{\braces{#1}}

\newcommand{\verts}[1]{\left\lvert #1 \right\rvert} 
\newcommand{\Verts}[1]{\left\lVert #1 \right\rVert} 
\newcommand{\abs}[1]{\verts{#1}}

\newcommand{\norm}[1]{\Verts{#1}}

\renewcommand{\Re}{\mathrm{Re}}
\renewcommand{\Im}{\mathrm{Im}}

{\begingroup  \begin{array}{l@{\hspace{8mm}}l@{\hspace{8mm}}l}}%
{\end{array} \endgroup}

\date{}
\begin{document}


\begingroup
\let\clearpage\relax
\title{Universal Connection Schedules for Reconfigurable Networking\thanks{This paper was accepted to the ACM-SIAM Symposium on Discrete Algorithms (SODA) 2026. This is a pre-print version.}}

\author{\email{Shaleen Baral}{sb969@cornell.edu}\thanks{Cornell University. $\{$sb969,rdk2,slm338,hwr34,rz234$\}$@cornell.edu} \and \email{Robert Kleinberg}{rdk@cs.cornell.edu}\samethanks \and \email{Sylvan Martin}{slm338@cornell.edu}\samethanks \and \email{Henry Rogers}{hwr34@cornell.edu}\samethanks \and \email{Tegan Wilson}{te.wilson@northeastern.edu}\thanks{Northeastern University. te.wilson@northeastern.edu} \and \email{Ruogu Zhang}{rz234@cornell.edu}\samethanks[1]
}

\begin{titlingpage}
\maketitle

\begin{abstract}

Reconfigurable networks are a novel communication paradigm in which the pattern of connectivity between hosts varies rapidly over time. 
Prior theoretical work explored the inherent tradeoffs between throughput (or, hop-count) and latency, and showed the existence of infinitely many Pareto-optimal designs as the network size tends to infinity.
Existing Pareto-optimal designs use a connection schedule which is fine-tuned to the desired hop-count $h$, permitting lower latency as $h$ increases.
However, in reality datacenter workloads contain a mix of low-latency and high-latency requests.
Using a connection schedule fine-tuned for one request type leads
to inefficiencies when serving other types.

A more flexible and efficient alternative is a {\em universal schedule}, a single connection schedule capable of attaining many Pareto-optimal tradeoff points simultaneously, merely by varying the choice of routing paths.
In this work we present the first universal schedules for oblivious routing. Our constructions yield universal schedules which are near-optimal for all possible hop-counts $h$.
The key technical idea is to specialize to a type of connection schedule based on cyclic permutations and to develop a novel Fourier-analytic method for analyzing randomized routing on these connection schedules.
We first show that a uniformly random connection schedule suffices with multiplicative error in throughput, and latency optimal up to a $\log N$ factor.
We then show that a more carefully designed random connection schedule suffices with additive error in throughput, but improved latency optimal up to only constant factors. 
This uses a tighter probabilistic bound pertaining to
martingales in Banach spaces.
Finally, we show that our first randomized construction can be made deterministic using a derandomized version of the Lovett-Meka discrepancy minimization algorithm to obtain the same result.

\end{abstract}

\end{titlingpage}

\section{Introduction} \label{sec:intro}

Reconfigurable networks are a novel communication paradigm in which the pattern of connectivity between hosts varies rapidly over time, creating a virtual topology that determines the paths through which data may be sent as links come in and out of existence. 
While practitioners are experimenting with these networks in an effort to unlock the energy savings and scalability benefits of replacing traditional packet switches with optical circuit switches~\cite{jupiter-rising-google-dc-network,sirius}, theoreticians have explored the design space of reconfigurable network architectures to understand the inherent tradeoffs between different design goals~\cite{optimal-orns,extending-optimal-orns,breaking-vlb}.
A key finding of prior theoretical work on reconfigurable networks is the existence of an unbounded number of distinct Pareto-optimal network designs, as the network size tends to infinity. 
These designs are parameterized by hop-count (i.e., average number of physical links per routing path), permitting progressively lower latency as the hop-count increases.
This is because unlike static networks, in reconfigurable networks each edge appears only at specific timesteps. 
Therefore, latency is defined as the number of timesteps one must wait for each intended physical hop to appear, rather than the number of hops taken. This is formalized in \Cref{sec:definitions}. 
As the allowable hop-count increases, this directly increases the number of unique destinations reachable within the same number of timesteps, allowing for lower-latency designs.
For workloads with completely homogeneous and predictable latency requirements, this allows network operators to choose the network design that meets their latency requirement using the lowest possible hop-count, which translates into the highest possible throughput.

In reality, datacenter workloads contain a mix of low-latency and high-latency connection requests, which creates additional challenges for the network designer. 
Existing Pareto-optimal reconfigurable network designs use connection schedules that are fine-tuned to optimize latency for a specific hop-count.
To serve heterogeneous workloads, the state of the art~\cite{shale} is to use time-sharing: connection requests are partitioned into subsets (often called \emph{traffic classes}) with homogeneous latency requirements and, correspondingly, the network's connection schedule is partitioned into sub-schedules each devoted to serving one of the traffic classes. 
Time-sharing uses the available bandwidth inefficiently, especially when the fraction of time slots devoted to each sub-schedule isn't perfectly matched to the fraction of connection requests served by that sub-schedule.

A more flexible and efficient way to do reconfigurable networking would be to use a \emph{universal connection schedule} capable of attaining many (or perhaps all) points on the Pareto frontier of throughput versus latency, merely by varying the choice of routing paths to minimize hop-count within the specified latency bound. 
The question of existence of universal connection schedules was raised as an open problem in prior theoretical work on reconfigurable networks~\cite{optimal-orns}. 
In this work, we present constructions of universal connection schedules which are near-optimal in the following respects.
\begin{description}
  \item[Multiplicative $\varepsilon$-approximation]
        First, we present a randomized construction that, for any network size $N$, achieves the following guarantee with high probability: for any feasible rate of guaranteed throughput, $r$, the connection schedule supports a routing scheme that guarantees throughput at least $(1-\varepsilon)r$, while the maximum latency of the routing paths exceeds the best possible latency by a factor of $O_{\varepsilon}(\log N)$.
  \item[Additive $\varepsilon$-approximation]
        Next, we present a different randomized construction that eliminates the $O(\log N)$ factor in the latency bound by relaxing the throughput guarantee from $(1-\varepsilon)r$ to $r-\varepsilon$. 
        This follows as a corollary of a more general result showing that for any finite set $R$ of throughput values, there is a universal connection schedule supporting routing schemes for each $r \in R$ that guarantee throughput $(1-\varepsilon)r$, while the maximum latency of the routing paths exceeds the best possible latency by a factor of $O_{\varepsilon,R}(1)$, i.e.~a constant factor that depends on both $\varepsilon$ and $R$ but not on the network size, $N$.
  \item[Deterministic construction]
        Finally, we show that the first aforementioned construction can be made deterministic and computationally efficient.
\end{description}

\subsection{Summary of results and techniques} \label{sec:techniques} 

A universal oblivious reconfigurable network (ORN) for a fixed network size $N$ consists of a single connection schedule and a family of oblivious routing protocols for each possible throughput value $r\in[0,\frac{1}{2}]$, each defined by routing paths on that connection schedule, and each of which is near-Pareto optimal.
Each of our universal ORN constructions is analyzed by making use of a novel connection between oblivious routing protocols, generating polynomials, and the norm of Fourier coefficient vectors. 

We standardize a single oblivious routing protocol that is well-defined on any connection schedule, only requiring two parameters: a hop count $h$ and a phase length $\Lambda$. 
This routing protocol uses exactly $2h$ hops per routing path, and attains maximum latency no more than $(2h+1)\Lambda$. 
However, it only guarantees throughput $\frac{1}{2h}(1-\eps)$ for some $\eps\geq 0$ which depends on properties of the connection schedule. 
To show that a single connection schedule is universal, we show that for each possible value of $h$, this oblivious routing protocol guarantees near-optimal throughput and latency for some carefully chosen phase length $\Lambda$ dependent on $h$.

Our oblivious routing protocol is motivated by the packet spraying implementation of Valiant Load Balancing (VLB), with a novel twist.
VLB is a technique which routes flow obliviously from source node $a$ to destination node $b$ by routing on shortest paths through a uniformly random intermediate node in the network.
VLB is provably optimal in some contexts, notably including reconfigurable networks~\cite{optimal-orns}.
In packet spraying, in order to reach a random intermediate node in $h$ hops, you send flow across $h$ {\em spraying hops} each chosen uniformly at random from a set of $\Lambda$ potential neighbors.
Then you route from that intermediate node to the destination node in $h$ hops, sampled using a time-reversal of the process used to select spraying hops.
In a graph with diameter $h$ and suitable expansion properties, this is approximately the same as routing on shortest paths through a uniformly random intermediate node.
In our oblivious routing protocol, we packet spray on $h$-hop paths, with exactly one hop chosen from consecutive non-overlapping phases each of length $\Lambda$.

One of our main innovations is departing from traditional VLB to allow the distribution on intermediate nodes to not be uniformly random.
We call this framework ``VLB with leakage.''
We show that if the distribution over intermediate nodes produced by packet spraying is not uniform, but instead only $\frac{\eps}{2}$-close to uniform, then packet spraying still guarantees throughput $\frac{1}{2h}(1-\eps)$.
However, the distribution over intermediate nodes depends on the underlying connection schedule that is used. 
Another main innovation in our work lies in designing universal connection schedules that enable $h$-hop packet spraying with an appropriate phase length to generate a nearly uniform distribution for many values of $h$ simultaneously.

As in prior work on reconfigurable network design~\cite{optimal-orns,breaking-vlb,shale}, we use a type of connection schedule in which the node set is identified with the elements of a finite abelian group, the time steps are identified with a sequence $\ldots,s_0, s_1, s_2, \ldots$ of group elements, and the node pairs that are connected at time $t$ are those whose difference equals $s_t$. 
We call such a structure a \emph{Cayley schedule}. 
It is analogous to a Cayley graph, except that each edge has an associated time step and paths can only be constructed from edges whose time steps are in increasing order. 
In prior work designing reconfigurable networks optimized for one specific throughput rate, the Cayley schedules were based on an $h$-dimensional vector space over a finite field, where the dimension $h$ was chosen depending on the desired rate of throughput. 
The generating sequence was then chosen by cycling through the elements of a vector space basis and their scalar multiples. 
Packet spraying using this design generates the uniform distribution simply because each element of the vector space is uniquely representable as a linear combination of basis elements. 
However, this design is clearly specialized to one value of $h$ and is not suitable as a universal connection schedule. 
Instead, in this work we focus on Cayley schedules defined over a \emph{cyclic} group. 
We refer to these as \emph{shift} connection schedules, since the connectivity pattern in any time step is represented by a cyclic shift permutation.

For shift connection schedules, we examine the generating polynomial of the packet-spraying distribution, a polynomial that is dependent on the shift connection schedule, $h$, and $\Lambda$.
  \begin{dfn} \label{def:fourier}
	Suppose $X$ is a $\zmodn$ valued random variable. Its \emph{generating polynomial} $p_X(z)$ and Fourier coefficient vector $\hat{p}_X[j]$ are given by
	\begin{equation} \label{eq:fourier} p_X(z) = \E[z^X] = \sum_{k = 0}^{N - 1} \Pr(X = k) \cdot z^k,
          \qquad \hat{p}_X[j] = p_X\left(e^{\frac{2\pi i j}{N}}\right) .
        \end{equation}
\end{dfn}

We show that if this Fourier coefficient vector of the packet-spraying distribution (modified by setting the $j=0$ coordinate to zero) has small 2-norm, this directly implies that packet spraying creates a distribution that is close to uniform.
We also show that this Fourier coefficient vector has some nice properties that make its 2-norm easier to bound.
It is equal to the element-wise product of averages of Fourier coefficients.
We explain this process in further detail in \Cref{sec:small-fourier-implies-orn}.
We exploit this connection to prove our three main results.

\begin{thm}\label{thm:main-thm}
	Let 
		\[ L^*(h,N) = hN^{1/h} . \]
	Given a constant $\eps>0$, 
	\begin{enumerate}
		\item \textbf{(Multiplicative error)} For sufficiently large $N$, there exists a randomized algorithm to compute a shift connection schedule $\bm{\sigma}$ on $N$ nodes which succeeds with probability at least $\frac{1}{2}$ and is efficiently certifiable.
		Additionally for each $h\in \{ 1,\ldots,\log_2N \}$, there also exists an oblivious routing protocol on $\bm{\sigma}$ which guarantees throughput $\frac{1}{2h}(1-\eps)$ and achieves maximum latency $\bigotilde(L^\star(h,N))$.
		\item \textbf{(Additive error)} For sufficiently large $N$, there exists a randomized algorithm to compute a shift connection schedule $\bm{\sigma}$ on $N$ nodes which succeeds with probability at least $\frac{1}{2}$ and is efficiently certifiable.
		Additionally for each $h\in \{ 1,\ldots,\log_2N \}$, there also exists an oblivious routing protocol on $\bm{\sigma}$ which guarantees throughput $\frac{1}{2h}-\eps$ and achieves maximum latency $\bigo(L^\star(h,N))$
		\item \textbf{(Deterministic, multiplicative error)} For infinitely many $N$, there exists a deterministic algorithm to compute a shift connection schedule $\bm{\sigma}$ on $N$ nodes which is efficiently computable.
			Additionally for each $h\in \{ 1,\ldots,\log_2N \}$, there also exists an oblivious routing protocol on $\bm{\sigma}$ which guarantees throughput $\frac{1}{2h}(1-\eps)$ and achieves maximum latency $\bigo(L^\star(h,N)\cdot \log(N))$.
	\end{enumerate}
\end{thm}

Before we summarize the main technical ideas we use to prove each part of \Cref{thm:main-thm}, two observations are in order.
  First, we restrict to hop-counts $h\in\{1,\ldots,\log_2 N\}$.
This is because for any fixed network size $N$, one can show that maximum latency cannot be improved beyond $\log_2 N$ using a simple connectivity argument; 
it is impossible to design a connection schedule that allows a node $a$ to reach every other node in less than $\log_2 N$ maximum latency.
There is an oblivious routing protocol that achieves this maximum latency by sending flow on routing paths that use up to $\log_2 N$ hops.
Therefore, increasing hop count beyond $\log_2 N$ can provide no further improvement in maximum latency.
Second, there is a lower bound~\cite{optimal-orns} showing $\Theta(L^*(h,N))$ is the optimal latency function for any guaranteed throughput value $\frac{1}{2h}(1-\eps)$ as long as $\eps\leq 1 - \frac{h}{h+1-(1/C)^h}$ for any constant, $C$.
$\Theta(L^*(h,N))$ is also the optimal latency function for any guaranteed throughput value $\frac{1}{2h}-\eps$ as long as $\eps \leq \frac{1}{2h}\left(1- \frac{h}{h+1-(1/C)^h}\right)$ for any constant, $C$.

To prove \Cref{thm:main-thm}.1, we consider a shift connection schedule where every shift is chosen independently and uniformly at random.
In this case, each individual Fourier coefficient of our generating polynomial can be decomposed into a sum of independently random Fourier coefficients.
We apply Hoeffding's inequality to bound each of these sumswith high probability, and then we apply a union bound over Fourier coefficients to derive a bound on the infinity-norm of the Fourier vector, from which a bound on its 2-norm follows. 
An additional union bound over all possible hop counts $h$ and starting timesteps $t$ in the connection schedule leads to the conclusion that the random shift construction produces a universal schedule with high probability.
We find that this analysis requires a phase length $\Lambda_h$ that is optimal up to $\bigo_\eps(\log N)$ factors, resulting in an overall latency bound of $\bigotilde_\eps(L^\star(h,N))$.
This analysis is further detailed in \Cref{sec:rand-shift-sched}.

For our second randomized result, \Cref{thm:main-thm}.2,
the main challenge is to remove the $\log(N)$ factor stemming from repeated application of the union bound. 
Two new ideas are required to accomplish this. 
First, rather than bounding the infinity-norm of the vector of Fourier coeffients, we bound its $q$-norm for a value of $q$ that is large but independent of $N$, using a dimension-independent martingale inequality for the Banach space $\ell^q$~\cite{pisier_2016}. 
This eliminates the union bound over Fourier coefficients. 
To eliminate the union bound over starting time steps and hop counts, we replace the purely random shift schedule used in \Cref{thm:main-thm}.1 with a more complicated construction. We introduce a \emph{convolution product} operation that combines two shift schedules into a longer schedule whose period length is the two product of their two periods. 
By starting with a purely random base schedule and forming higher and higher powers using the convolution product, we show using H\"{o}lder's Inequality that bounds on the $q$-norm of the base schedule's Fourier coefficient vector (for a suitable choice of $q$) imply bounds on the 2-norm of the Fourier coefficients for the resulting convolved schedule. 
Bounds on the Fourier 2-norm for every possible starting timestep $t$, and for any finite number of different hop-counts $h$, can thereby be derived from a single Fourier $q$-norm bound for the base schedule, avoiding the use of the union bound.
As a corollary, we show this implies additive $\eps$-approximation of throughput for all possible hop counts $h$, with the correct setting of the base schedule's period length. 
This analysis is further detailed in \Cref{sec:convolution-sched}.

Finally to prove \Cref{thm:main-thm}.3, we show a connection between designing a universal shift connection schedule and the high-dimensional vector discrepancy problem.
We define a deterministic shift connection schedule by recursively splitting the set of all Fourier coefficient vectors into equal-size sets, maintaining that after each split, the sum of vectors in a set has $\infty$-norm only slightly larger than the $\infty$-norm of the sum of vectors before the split.
This turns out to be possible by relying on a deterministic version of the Lovett-Meka discrepancy minimization algorithm \cite{deterministic-discrepancy}, and directly implies a bound on the $\infty$-norm of the relevant Fourier coefficient vector for $h$-hop packet spraying.
Then for each value of $h$, a careful setting of $\Lambda=\bigo_\eps(N^{1/h}\cdot\log N)$ ensures the resulting distribution is $\frac{\eps}{2}$-close to uniform.
This analysis is further detailed in \Cref{sec:derandomized}.

\subsection{Related work} \label{sec:related-works}

\textbf{Latency-Throughput Tradeoffs in Reconfigurable Networks.} 
The most important related works \cite{optimal-orns,extending-optimal-orns,breaking-vlb} are mentioned above in \Cref{sec:intro}.
Recent work formally defined oblivious reconfigurable networks (ORNs), and showed the existence of designs that achieve Pareto-optimal tradeoffs between throughput and latency.
That is, for any fixed throughput rate $r$, they constructed a family of infinitely many ORN designs which guaranteed throughput $r$ and achieved optimal maximum latency, up to a constant factor~\cite{optimal-orns}.
These results were later extended to any sufficiently large network size~\cite{extending-optimal-orns}.
Additionally the Pareto-frontier was significantly improved upon with the introduction of random ORN designs whose guarantees are relaxed to only hold with high probability~\cite{breaking-vlb}.
Each of these designs uses a connection schedule that is tailored to the tradeoff point they wish to guarantee.
Our work aims to construct schedules that capture not just one tradeoff point, but multiple, at the expense of an approximate guarantee.

\textbf{Valiant Load Balancing (VLB).} 
Valiant and Brebner first introduced oblivious routing in their seminal work~\cite{vlb-valiant81,valiant82}, demonstrating that 
their strategy, which became known as Valiant Load Balancing (VLB),
achieves optimal latency guarantees over the direct-connect hypercube topology. 
VLB is the optimal oblivious routing scheme
in at least two other contexts: when building a network of fixed-capacity links~\cite{keslassy2005optimal,zhang2005designing,babaioff2007optimality}, and in ORNs~\cite{optimal-orns}, in both cases when aiming to fulfill any demand pattern with bounded ingress and egress rates per node.
Our work builds upon these results, formalizing
a relaxation of VLB that enables guaranteeing approximately optimal  
throughput. 

\textbf{Practical ORN Proposals.} 
Many proposed data center architectures~\cite{rotornet,sirius,shoal,mars,shale} fit the ORN paradigm, some even before it was formally introduced~\cite{optimal-orns}. 
Rotornet~\cite{rotornet} and Sirius~\cite{sirius} propose using optical circuit switches to implement the reconfigurable network fabric, while Shoal~\cite{shoal} proposes using electrical circuit switches.
These three designs all attain the same single Pareto-optimal tradeoff point, prioritizing high throughput at the expense of poor latency.
Shale~\cite{shale} uses connection schedules from theoretical work~\cite{optimal-orns}, allowing for a tunable tradeoff between throughput and latency, that notably also improved upon resource scaling from previous designs.
Mars~\cite{mars} takes an entirely different approach, motivated by the setting of an ORN with limited buffer memory at each node.
Given an arbitrary amount of buffer space, they devise a formula for the ideal emulated degree (or, period) of the schedule, and emulate a de Bruijn graph with this ideal degree.
Additionally, more recent work~\cite{semi-obliv-hotnets} proposes a semi-adaptive reconfigurable connection schedule, relying on the predictability of large-scale structural patterns in datacenter traffic.

\textbf{Oblivious Routing in General Networks.} 
There is a line of work in oblivious routing primarily focused on minimizing congestion in general networks, compared to an optimal offline routing algorithm. 
The existence of an oblivious online routing algorithms that obtain a polylogarithmic competitive ratio was initially proven by R{\"a}cke \cite{raecke02}. 
Later work improved upon this competitive ratio~\cite{harrel03},
and described polynomial-time constructions that obtain these bounds\cite{azar03,bienk03,harrel03}.
Using multiplicative weights and tree metrics techniques introduced by \cite{frt04}, R\"acke further improved the competitive ratio to $O(\log n)$~\cite{raecke08}, which was closer to the previously known lower bound of $\Omega(\frac{\log n}{\log \log n})$~\cite{hajia06}. 
This oblivious routing algorithm was later combined with centralized traffic engineering solutions to produce a robust and near-optimal system~\cite{smore18}.
Although congestion, the focus of much of this line of work, and throughput can be direct inverses of each other, their definition of traffic demands (and therefore throughput) is different from ours, which we take from previous work in reconfigurable network theory~\cite{optimal-orns,extending-optimal-orns,breaking-vlb}.
Additionally, we aim to create an {\em optimized topology} to do oblivious routing on, rather than designing algorithms that perform well on an arbitrary topology.

\section{Definitions} \label{sec:definitions}

\begin{dfn} \label{def:connection-schedule}
	A {\em connection schedule} of $N$ nodes and period length $T$ is a sequence of permutations $\bm{\sigma}=\sigma_0,\sigma_1,\ldots,\sigma_{T-1}$, each mapping $[N]$ to $[N]$. 
	$\sigma_k(i) = j$ means that node $i$ is allowed to send one unit of flow to node $j$ during any timestep $t$ such that $t \equiv k \pmod{T}$.

	The {\em virtual topology} of the connection schedule $\bm{\sigma}$ is a directed graph $\gpi$ with vertex set $[N] \times \mathbb{Z}$.
	The edge set of $\gpi$ is the union of two sets of edges, $\evirt$ and $\ephys$.
	$\evirt$ is the set of {\em virtual edges}, which are of the form $(i,t)\to(i,t+1)$ and represent flow waiting at node $i$ during the timestep $t$.
	$\ephys$ is the set of {\em physical edges}, which are of the form $(i,t)\to(\sigma_t(i),t+1)$, and represent flow being transmitted from $i$ to $\sigma_t(i)$ during timestep $t$. 
	A useful automorphism of $\gpi$ is \emph{time-translation}, the function defined on vertices of $\gpi$ by $\tau(a,t) = \tau(a,t+T)$. 
	This function maps edges to edges because the connection schedule is $T$-periodic.
\end{dfn}

\begin{dfn} \label{def:shift-connection-schedule}
	Interpret the set of $N$ nodes to be the commutative group $G = \mathbb{Z} / (N)$ under addition. 
	A \emph{shift connection schedule} on $G$ is a connection schedule where the permutations $\bm{\sigma}=\sigma_0,\sigma_1,\ldots,\sigma_{T-1}$ are all defined by shifts $s_0, \ldots, s_{T-1} \in G$, where $\sigma_k(j) = j + s_k$ for all $k \in \{0, \ldots, T-1\}$.
\end{dfn}

We interpret a path in $\gpi$ from $(a,t)$ to $(b,t^\prime)$ as a potential way to transmit one unit of flow from node $a$ to node $b$, beginning at timestep $t$ and ending at some timestep $t^\prime>t$.
Let $\pths$ denote the set of all paths in $\gpi$. Some useful subsets of $\pths$ are the following.
	\begin{align*}
		\pths_{a,b,t,t'}                                                           & = \{\mbox{paths starting at $(a,t)$ and ending at $(b,t')$}\} \\
		\pths_{a,b,t}                                                              & = \bigcup_{t' > t} \pths_{a,b,t,t'}                           \\
		\pths_{(a,t) \rightarrow}  = \bigcup_{b \in [N], t > t'} \pths_{a,b,t,t'}, & \qquad
		\pths_{\rightarrow (b,t')} = \bigcup_{a \in [N], t < t'} \pths_{a,t,b,t'}                                                                  \\
		\pths^e                                                                    & = \{\mbox{paths containing edge } e\} .
	\end{align*}
	The time-translation automorphism $\tau : \gpi \to \gpi$ induces an
	automorphism of the path set, which we will also denote by $\tau$.

\begin{dfn} \label{def:flow}
	A {\em flow} is a measure on the measurable space $(\pths,2^\pths)$.
		A flow $f$ is called $T$-periodic if it is invariant under time-translation, i.e.~$f(\tau(P)) = f(P)$ for all $P$.
		For a flow $f$ and edge $e$, the amount of flow traversing $e$ is defined to be $f(\pths^e)$. We say that $f$ is {\em feasible} if the amount of flow tranversing every \emph{physical} edge is at most 1. Note that in our definition of feasible, we allow virtual edges to have unlimited capacity.
\end{dfn}

\begin{dfn}\label{def:routing-scheme}
	An {\em oblivious routing protocol} $R$ is a $T$-periodic flow that defines, for each triple $(a,b,t) \in [N] \times [N] \times \mathbb{Z}$, a unit flow from the starting node $(a,t)$ to the node set $\{(b,t') \mid t' > t \}$. 
	In other words, $R$ satisfies $R(\pths_{a,b,t})=1$ for all $(a,b,t) \in [N] \times [N] \times \mathbb{Z}$.
\end{dfn}

\begin{dfn}
		A {\em forward routing distribution} $\mathcal{D}$ is a $T$-periodic flow that defines, for every starting node and timestep $(a,t)\in[N]\times\Z$, a unit of flow supported on $\pths_{(a,t) \rightarrow}$.
		In other words, $\mathcal{D}(\pths_{(a,t) \rightarrow}) = 1$ for all $(a,t)\in[N]\times\Z$.
		Then $\mathcal{D}_{a,t}(b) = \sum_{P\in \pths_{a,b,t}} \mathcal{D}(P)$ denotes the amount of flow assigned to paths which start at $(a,t)$ and end at node $b$ during any ending timestep.
		Note that for any fixed $(a,t)\in[N]\times\Z$, $\mathcal{D}_{a,t}(b)$ forms a probability distribution over destination nodes $b$.

		A {\em backwards routing distribution} $\mathcal{E}$ is a $T$-periodic flow that defines, for every destination node and ending timestep $(b,t')\in[N]\times\Z$, a unit of flow supported on $\pths_{\rightarrow (b, t')}$.
		In other words, $\mathcal{E}(\pths_{\rightarrow (b, t')}) = 1$ for all $(b,t')\in[N]\times\Z$.
		Then $\mathcal{E}_{b,t'}(a) = \sum_{t}\sum_{P\in \pths_{a,b,t,t'}} \mathcal{E}(P)$ denotes the amount of flow assigned to paths which start at node $a$ during any timestep, and end at node $b$ and ending timestep $t'$.
		Note that for any fixed $(b,t')\in[n]\times\Z$, $\mathcal{E}_{b,t'}(a)$ forms a probability distribution over starting nodes $a$. 
\end{dfn}

\begin{dfn} \label{def:latencies}
	The {\em latency} $L(P)$ of a path $P$ in $\gpi$ is equal to the number of edges it contains (both virtual and physical).
	Traversing any edge in the virtual topology (either virtual or physical) is equivalent to advancing in time by one timestep, so the number of edges in a path equals the elapsed time.
	For a flow $f$ the {\em maximum latency} is the maximum over all paths $P$ which may route flow:
	\begin{align*}
		L_{max}(f) & = \max \{ L(P) \, \mid \, f(P) > 0 \} .
	\end{align*}
\end{dfn}

\begin{dfn} \label{def:demand-functions}
	A {\em demand function} $D$ is a function that associates to every source-destination pair $(a,b) \in [N] \times [N]$ and time step $t \in \mathbb{Z}$ a quantity $D(a,b,t) \geq 0$ representing the amount of flow that $a$ aims to send to $b$ starting at time $t$.

	If $f$ is a flow and $D$ is a demand function, the {\em induced flow} $f^D$ is defined by rescaling $f$ using $D(a,b,t)$ as the scale factor for paths in $\pths_{a,b,t}$:
		\[
			f^D(P) = D(a,b,t) \cdot f(P) \quad \forall P \in \pths_{a,b,t} .
		\]
		We will typically be applying this rescaling operation when $f$
		is a routing protocol or a (forward or backward) routing distribution.
\end{dfn}

\begin{dfn} \label{def:guaranteed-thr}
	A routing protocol $R$ {\em guarantees throughput $r$} if the induced flow $R^D$ is feasible whenever $D$ is a demand function satisfying
		\[ \forall a \in [N], \, t \in \mathbb{Z} \quad
			\sum_{b \in [N]} D(a,b,t) \leq r \quad \mbox{and} \quad
			\sum_{b \in [N]} D(b,a,t) \leq r .
		\] 
	An equivalent definition~\cite{optimal-orns} requires the induced flow $R^D$ to be feasible for any $D$ that is a set of permutation demands scaled by $r$. That is, any $D$ for which, for all $t\in\Z$, there exists a permutation $\tau_t:[N]\rightarrow[N]$ such that 
	\begin{align*}
		& \forall a \in [N], \, t \in \mathbb{Z}, D(a,\tau_t(a),t) = r \\
		& \forall a \in [N], \, t \in \mathbb{Z}, b\neq\tau_t(a), D(a,b,t) = 0 .
	\end{align*}	
\end{dfn}

\section{Reducing Connection Schedule Design to Fourier Analysis}\label{sec:small-fourier-implies-orn}

This section is devoted to proving the following theorem, which relates universal connection schedule design to Fourier analysis.

\begin{thm}\label{thm:small-fourier-implies-orn}
	Given a shift schedule $\bm{\sigma} = s_0,s_1,\ldots,s_{T-1}$ on $G = \zmodn$, and given parameters $\Lambda$ and $h$ for which $T\geq \Lambda h$.

	If $\Lambda h$ evenly divides $T$, then let $\mathcal{A} = \{\ell \Lambda h : \ell \in \mathbb{Z}_{\geq 0}\}$. Otherwise, let $\mathcal{A} = \{0,\ldots,T-1\}$.
	Then for each timestep $t \in \mathcal{A},$
	let $p_t$ and $q_t$ be the polynomials
	\[ p_t(z) = \prod_{j=1}^{h} p_{t,j}(z) , \hspace{5mm} q_t(z) = \prod_{j=h+1}^{2h} q_{t,j}(z) \]
	where
	\[ p_{t,j}(z) = \frac{1}{\Lambda} \left( \sum_{k= t + (j-1)\Lambda}^{t + j\Lambda-1} z^{s_{k}} \right) , \hspace{5mm} q_{t,j}(z) = \frac{1}{\Lambda} \left( \sum_{k=t+(j-1)\Lambda}^{t+j\Lambda-1} z^{N-s_{k}} \right) \]
	Define $\hat{p_t}^\star$ as the Fourier transform of $p$ with first coordinate set to 0, and $\hat{q_t}^\star$ similarly for $q$.

	Now let $\eps>0$ be a constant.
	If $\norm{\hat{p_t}^\star}_2 \leq \frac{\eps}{2}$ and $\norm{\hat{q_t}^\star}_2 \leq \frac{\eps}{2}$ for all $t\in\mathcal{A}$, then there exists an oblivious routing protocol over $\bm{\sigma}$ which guarantees throughput $\frac{1}{2h}(1 - \eps)$ and maximum latency no more than $2(h+1)\Lambda$.
\end{thm}

This section describes our $h$-spraying-hop-inspired oblivious routing protocol, which we call {\em Valiant Load Balancing with Leakage}.
We first show that this oblivious routing protocol guarantees throughput $\frac{1}{2h}(1-\eps)$ when the forward and backwards spraying hop distributions used to define it are $\frac{\eps}{2}$-close to uniform.
We then devote some time to introducing the Fourier analysis tools used in this paper.
Finally, we connect showing that our spraying hop distributions are $\frac{\eps}{2}$-close to uniform to showing that $\norm{\hat{p_t}^\star}_2 \leq \frac{\eps}{2}$.

\subsection{Valiant Load Balancing with Leakage}\label{sec:vlb-w-leakage}

Let $\mathcal{D}$ be a forward routing distribution.
If we could show that $\mathcal{D}_{a,t}$ forms the uniform distribution over destination nodes for all $a,t$, then normalizing $\mathcal{D}(P)$ by a factor of $N$ would give an oblivious routing protocol on $\bm{\sigma}$, though not necessarily one with good throughput guarantees.
However, if $\mathcal{D}$ was good at routing uniform traffic demands, then we could use the following construction to define an oblivious routing protocol $R$ with certain throughput and latency guarantees.

\begin{dfn}[\textbf{Valiant Load Balancing \cite{vlb-valiant81}}]\label{def:vlb}
	Given an oblivious routing protocol $\mathcal{U}$ with maximum latency $L$ on a $N$-node connection schedule, we can define a second oblivious routing protocol, $VLB(\mathcal{U})$, motivated by the following informal specification. To route between source-destination pair $(a,b)$ starting at timestep $t$,
	\begin{enumerate}
		\item Traffic is routed using $\mathcal{U}$ from $(a,t)$ to a uniformly random intermediate node $c$ in the network.
		\item Traffic is routed from the intermediate node $(c, t + L)$ to its final destination $b$.
	\end{enumerate}
	We call these two pieces of the routing {\em semi-paths} of the oblivious routing protocol.

	While the informal specification of the routing protocol refers to choosing exactly one intermediate node at random, within the formalism of \Cref{def:flow} this would be encoded using a flow which averages over the random choice of intermediate node $c$. Thus, if $P_0 \in \pths_{a,c,t,t+L}$ and $P_1 \in \pths_{c,b,t+L,t+2L}$ are two paths and $P = P_0 \pathcat P_1$ denotes their concatenation, then $$VLB(\mathcal{U})(P) = \frac1N \mathcal{U}(P_0) \mathcal{U}(P_1).$$
\end{dfn}

\begin{thm}[\textbf{Valiant Load Balancing \cite{vlb-valiant81}}]\label{thm:vlb}
	Let $D_{unif}$ be the uniform traffic demand function with all demands equal to $\frac{1}{N}$.
	Let $\mathcal{U}$ be an oblivious routing protocol on an $N$-node connection schedule with maximum latency $L$ such that the induced flow $f^{\frac{1}{h}D_{unif}}$ is a feasible flow.
	Then $VLB(\mathcal{U})$ is an oblivious routing protocol that achieves maximum latency $2L$ and guarantees throughput $\frac{1}{2h}$.
\end{thm}

Requiring $\mathcal{D}$ to form the uniform distribution on destination nodes is quite restrictive.
However, we will show that even a distribution that is close to uniform 
will be good enough to apply something akin to Valiant Load Balancing (VLB).

\begin{dfn}[\textbf{$h$-hop spraying forward and backward routing distributions}]\label{def:h-hop-spraying}
	~\\
	\indent Let $\bm{\sigma}$ be a shift connection schedule on $\zmodn$ defined by shifts $s_0,s_1,\ldots,s_{T-1}$.
	In addition, suppose $T$ is divisible by $h$ for hop parameter $h$, and let $\Lambda=T/h$.
	Partition $\bm{\sigma}$ into $h$ consecutive sub-sequences, each of length $\Lambda$.
	\begin{align*}
		\bm{\sigma}_1   & = s_0,\ldots,s_{\Lambda-1}          \\
		\bm{\sigma}_2   & = s_{\Lambda},\ldots,s_{2\Lambda-1} \\
		                & \vdots                              \\
		\bm{\sigma}_{h} & = s_{(h-1)\Lambda},\ldots,s_{T-1}
	\end{align*}
	We define the forward routing distribution $\mathcal{D}$ given by $h$-hop packet spraying on phase length $\Lambda$.
	For every path $P$ that first delays flow until the start of a period, then takes exactly one hop from each of the $h$ subsequences $\bm{\sigma}_1,\ldots,\bm{\sigma}_h$, set $\mathcal{D}(P) = \left(\frac{1}{\Lambda}\right)^{h}$.

	We also define the backwards routing distribution $\mathcal{E}$ formed by $h$-hop packet spraying in a time-reversed process.
	Ending at timestep $t'$, let $t^\star$ be the latest timestep that started a period\footnote{In a longer connection schedule, $\mathcal{E}$ paths may be defined over a different set of sub-sequences $\bm{\sigma}_{\ell},\ldots,\bm{\sigma}_{\ell+h}$. This is without loss of generality.} no later than $t'$.
	For every path $P$ which starts at timestep $t^\star-T$, takes exactly one hop from each of the $h$ subsequences $\bm{\sigma}_1,\ldots,\bm{\sigma}_h$, and then delays flow at its destination until timestep $t$, set $\mathcal{E}(P) = \left(\frac{1}{\Lambda}\right)^{h}$.
\end{dfn}

\begin{dfn} [\textbf{VLB with Leakage}] \label{def:vlb-w-leakage}
	Given a forward routing distribution $\mathcal{D}$ with maximum latency $L_\mathcal{D}$ and a backwards routing distribution $\mathcal{E}$ with maximum latency $L_\mathcal{E}$, we define the following strategy, which we call {\em Valiant Load Balancing with Leakage}, which concatenates them into a single oblivious routing protocol with maximum latency $L = L_\mathcal{D}+L_\mathcal{E}$.

	Consider a fixed starting node $a$, destination node $b$, and starting timestep $t$.
	To define an oblivious routing protocol $R$, we must define a unit of flow supported on $\pths_{a,b,t}$.
	We will further restrict ourselves to defining a unit of flow supported on $\pths_{a,b,t,t'}$ where $t'=t+L_\mathcal{D}+L_\mathcal{E}$.

	Now consider the probability distributions $\mathcal{D}_{a,t}(c)$ and $\mathcal{E}_{b,t'}(c)$.
	For each intermediate node $c$, we will first define how to send $\min\{\mathcal{D}_{a,t}(c), \mathcal{E}_{b,t'}(c)\}$ flow from $a$ to $b$ through intermediate node $c$, on paths concatenated from paths that support $\mathcal{D}$ and $\mathcal{E}$.

	To do so, let $t'' = t + L_{\mathcal{D}}$ and find any fractional maximum weighted matching $M$ between $\{P_1\in\pths_{a,c,t,t''} : \mathcal{D}(P_1) > 0 \}$ and $\{P_2\in\pths_{c,b,t'',t} : \mathcal{E}(P_2) > 0 \}$, where the edge weights are defined by $w(P_1,P_2) = \min\{ \mathcal{D}(P_1), \mathcal{E}(P_2) \}$.
	Then $M$ has total weight $\min\{\mathcal{D}_{a,t}(c), \mathcal{E}_{b,t'}(c)\}$.\footnote{One way to justify this claim is as follows. Subdivide the interval $[0,\mathcal{D}_{a,t}(c)]$ into subintervals $I(P)$ with disjoint interiors, indexed by paths $P \in \pths_{a,c,t,t''}$, with $I(P)$ having length $\mathcal{D}(P)$. Similarly, subdivide the interval $[0,\mathcal{E}_{b,t'}(c)]$ into subintervals $J(P)$ with disjoint interiors, indexed by paths $P \in \pths_{c,b,t'',t'}$, with $J(P)$ having length $\mathcal{E}(P)$. Then, define $M(P_1,P_2)$ to equal the length of $I(P_1) \cap J(P_2)$.}
	For each pair of paths $P_1$ in the first set and $P_2$ in the second set, let $R(P_1 \pathcat P_2) = M(P_1,P_2)$ be the amount of flow assigned to the concatenation\footnote{Because $P_1$ ends at $(c,t+L_\mathcal{D})$ and $P_2$ starts at the same place, this concatenated path is well-defined.} of $P_1$ and $P_2$.

	In total, we have defined a function which assigned $\sum_{c} \min\{\mathcal{D}_{a,t}(c), \mathcal{E}_{b,t'}(c)\}$ flow on routing paths from node $a$ to node $b$ starting at timestep $t$.
	Because $\mathcal{D}_{a,t}(c)$ and $\mathcal{E}_{b,t'}(c)$ are probability distributions, this sum is always no more than $1$.
	If it is not equal to 1, then normalize\footnote{The normalization operation is undefined when $\sum_{c} \min\{\mathcal{D}_{a,t}(c), \mathcal{E}_{b,t'}(c)\}=0$, i.e.~when the distributions $\mathcal{D}_{a,t}$ and $\mathcal{E}_{b,t'}$ are supported on disjoint sets of nodes. In all of the constructions we propose in this paper, the total variation distance between $\mathcal{D}_{a,t}$ and $\mathcal{E}_{b,t'}$ is strictly less than 1, so they cannot be disjointly supported.} $R(P)$ for $P\in\pths_{a,b,t}$ by this value.
\end{dfn}

\begin{lem}[\textbf{VLB with Leakage}]\label{lem:vlb-w-leakage}
	Let $\mathcal{D}$ and $\mathcal{E}$ be the $h$-hop spraying forward and backwards routing distributions for parameters $h$ and $\Lambda$, as defined in \Cref{def:h-hop-spraying}.
	Additionally, suppose for every pair of nodes $a,b$, and pairs of timesteps $t,t'$, $\mathcal{D}$ and $\mathcal{E}$ form probability distributions over source/destination nodes that are both $\frac{\eps}{2}$-close to the uniform distribution over $N$ nodes.
	That is,
	\begin{align*}
		\sum_c \abs{\mathcal{D}_{a,t}(c) - \frac{1}{N} }   & \leq \frac{\eps}{2}   \\
		\sum_c \abs{ \mathcal{E}_{b,t'}(c) - \frac{1}{N} } & \leq \frac{\eps}{2} .
	\end{align*}

	Then the routing protocol defined by applying the VLB with Leakage construction to $\mathcal{D}$ and $\mathcal{E}$ achieves maximum latency $2(h+1)\Lambda$ and guarantees throughput $\frac{1}{2h}(1-\eps)$.
\end{lem}

\begin{proof}
	Let $\vlbleak(\mathcal{D},\mathcal{E})$ be the oblivious routing protocol formed by applying VLB with leakage on $\mathcal{D}$ and $\mathcal{E}$.
	Since $\mathcal{D}$ and $\mathcal{E}$ both have maximum latency $(h+1)\Lambda$, then clearly $\vlbleak(\mathcal{D},\mathcal{E})$ will have maximum latency $2(h+1)\Lambda$.

	To compute its guaranteed throughput, consider that $\mathcal{D}$ and $\mathcal{E}$ evenly load balance edges in the network and use $h$ hops each per routing path.
	That is, define the demand function $D(a,c,t) = \frac{1}{h}\mathcal{D}_{a,t}(c)$.
	Then the induced flow $\mathcal{D}^D$ is a feasible flow.
	Similarly, define the demand function $E(c,b,t) = \frac{1}{h}\mathcal{E}_{b,t'-2(h+1)\Lambda}(c)$.
	Then the induced flow $\mathcal{E}^E$ is also a feasible flow.

	If we want to route $\mathcal{D}$ and $\mathcal{E}$ at the same time without overloading edge capacities, then we can do so by halving the amount of total flow that each routes.
	That is, the additive induced flow $\mathcal{D}^{\frac{1}{2}D} + \mathcal{E}^{\frac{1}{2}E}$, where the amount of flow crossing an edge is the sum of both flows crossing the edge, is also a feasible flow.

	When we apply the VLB with Leakage construction, we concatenate $\mathcal{D}$ and $\mathcal{E}$ path together, and before normalization, define a function $\unvlbleak(\mathcal{D},\mathcal{E})$ that assigns $\sum_{c} \min\{\mathcal{D}_{a,t}(c), \mathcal{E}_{b,t'}(c)\} \leq 1$ flow total on routing paths from node $a$ to node $b$ starting at timestep $t$.
	$\unvlbleak(\mathcal{D},\mathcal{E})$ is not an oblivious routing protocol, because it does not normalize the amount of flow that can be routed between any source-destination pair $a,b$ at timestep $t$ to be 1.
	However, for any permutation demand function $F$ defined by permutations $\tau_t:[N]\rightarrow[N]$,
	the induced flow $\unvlbleak(\mathcal{D},\mathcal{E})^{\frac{1}{2h}F} \leq \mathcal{D}^{\frac{1}{2}D} + \mathcal{E}^{\frac{1}{2}E}$, which is a feasible flow.

	So, what is the minimum amount of flow, $\eta$, that $\unvlbleak(\mathcal{D},\mathcal{E})$ tells us how to route?
	Once we lower bound this value, we can use it to derive the guaranteed throughput of the oblivious routing protocol, $\vlbleak(\mathcal{D},\mathcal{E})$.

	To define $\eta$ more precisely, let
	\begin{align*}
		\eta = \min_{a,b,t}\Big\{ \sum_c \min\{ \mathcal{D}_{a,t}(c), \mathcal{E}_{b,t+4T}(c)\} \Big\}.
	\end{align*}
	Since $\mathcal{D}$ and $\mathcal{E}$ are $T$-periodic for a finite $T$, then there exists an $(a,b,t)\in[N]^2\times\Z$ for which $\eta = \sum_c \min\{ \mathcal{D}_{a,t}(c), \mathcal{E}_{b,t+4T}(c)\}$.
	Fix such an $(a,b,t)$.
	Then the following holds.
	\begin{align*}
		\sum_c \mathcal{D}_{a,t}(c) & = 1 = \sum_c \mathcal{E}_{b,t+4T}(c)                                                                                                 \\
		\implies 2                  & = \sum_c \mathcal{D}_{a,t}(c) + \sum_c \mathcal{E}_{b,t+4T}(c)                                                                       \\
		                            & = \sum_c \min\{ \mathcal{D}_{a,t}(c), \mathcal{E}_{b,t+4T}(c) \} + \sum_c \max\{ \mathcal{D}_{a,t}(c), \mathcal{E}_{b,t+4T}(c) \}    \\
		                            & = \sum_c \big( 2\min\{ \mathcal{D}_{a,t}(c), \mathcal{E}_{b,t+4T}(c) \} + \abs{\mathcal{D}_{a,t}(c) - \mathcal{E}_{b,t+4T}(c)} \big) \\
		\implies 2                  & = 2\eta + \sum_c \abs{\mathcal{D}_{a,t}(c) - \mathcal{E}_{b,t+4T}(c)}
		= 2 \eta + \norm{\mathcal{D}_{a,t} - \mathcal{E}_{b,t+4T}}_{1} .
	\end{align*}
	Note that we can also bound 
	$\norm{\mathcal{D}_{a,t} - \mathcal{E}_{b,t+4T}}_{1}$
	by relying on the fact that $\mathcal{D}$ and $\mathcal{E}$ are $\frac{\eps}{2}$-close to uniform.
	\begin{align*}
		\norm{\mathcal{D}_{a,t} - \mathcal{E}_{b,t+4T}}_{1} & \leq \norm{\mathcal{D}_{a,t} - \bm{u}}_1 + \norm{\bm{u} - \mathcal{E}_{b,t+4T}}_1 \\
		 & = \sum_c \abs{\mathcal{D}_{a,t}(c) - \frac{1}{N} } +
		\sum_c \abs{ \mathcal{E}_{b,t'}(c) - \frac{1}{N} }
		\leq \eps ,
	\end{align*}
	where $\bm{u}$ is the uniform distribution over $N$ nodes.
	Plugging everything in, we see that
	\begin{align*}
		2 & \leq 2\eta + \eps
		\implies 	\eta \geq 1 - \eps .
	\end{align*}
	Recall that $\unvlbleak(\mathcal{D},\mathcal{E})^{\frac{1}{2h}F}$ is a feasible flow for any permutation demand function $F$.
	Then $\frac{1}{\eta}\unvlbleak(\mathcal{D},\mathcal{E})^{\frac{\eta}{2h}F}$ must also be a feasible flow.
	And finally, as a flow, $\frac{1}{\eta}\unvlbleak(\mathcal{D},\mathcal{E}) \geq  \vlbleak(\mathcal{D},\mathcal{E})$ since $\eta$ is a lower bound on the values that we normalize by to compute $\vlbleak$.
	Therefore, the routing protocol $\vlbleak(\mathcal{D},\mathcal{E})$ guarantees throughput $\frac{\eta}{2h} \geq \frac{1}{2h}(1-\eps)$.

\end{proof}

\subsection{Connections to Fourier Analysis}
\label{sec:fourier-connections}

In this subsection we review basic definitions and facts
about the discrete Fourier transform. The reader familiar with
this material may skip directly to the proof of
\Cref{thm:small-fourier-implies-orn} at the end of
the subsection.

In order to show that our forward and backwards routing distributions are close to uniform, we make use of Fourier analysis on the generating polynomials of our routing distributions. Recall the generating polynomial and Fourier coefficient vector of a distribution on $\zmodn$, defined in Equation~\eqref{eq:fourier}.
For independent, integer-valued random variables, the convolution of their generating functions completely characterizes the distribution of their sum.
A similar result holds for $\zmodn$ valued independent random variables.

\begin{thm} 
	Suppose $X, Y$ are independent, $\zmodn$ valued random variables. Then
	\begin{align*}
		p_{X + Y} = p_X p_Y \mod {z^N - 1} ,
		\hspace{3mm}\mbox{ {\normalfont and} }\hspace{3mm}
		\hat{p}_{X + Y} = \hat{p}_X \odot \hat{p}_Y ,
	\end{align*}
	where $\bm{v} \odot \bm{w}$ is the entry-wise product of the vectors $\bm{v}$ and $\bm{w}$.
\end{thm}

\begin{proof}
	By definition, for $k \in \left[ 0, N - 1 \right]$,
	\begin{align*}
		\Pr(X + Y \equiv k \bmod{N}) = \sum_{\substack{i + j \equiv k \\ \bmod{N}}} \Pr(X = i) \Pr(Y = j)
	\end{align*}
	Note that as $i, j \in [0, N - 1]$, there are two possibilities for each element of the sum on the right: either $i + j = k$ or $i + j = k + N$. However, we have $z^k \equiv z^{k + N} \bmod{z^N - 1}$ which follows from noting that $z^{k + N} - z^k = z^k \left( z^N - 1 \right)$.
	Thus,
	\begin{align*}
		\Pr(X + Y \equiv k \bmod{N}) \cdot z^k \hspace{10pt} \equiv_{(z^N - 1)} \sum_{\substack{i + j \equiv k \\ \mod{N}}} \Pr(X = i) \Pr(Y = j) \cdot z^{i + j},
	\end{align*}
	from which the first desired result follows.

	The $j$th index of the Fourier coefficients of $p_{X + Y}$ is
	\[
		\hat{p}_{X + Y}[j] = p_{X + Y}(\xi^j) = p_X(\xi^j) p_Y(\xi^j) = \hat{p}_X[j] \hat{p}_Y[j]
	\]
	where $\xi = e^{2\pi i/N}$.
	Therefore, $\hat{p}_{X + Y} = \hat{p}_X \odot \hat{p}_Y$.
\end{proof}

\begin{dfn}
	Let $\mathcal{P}\left( d \right)$ denote the vector space of polynomials over $\C$ with degree \emph{strictly} less than $d$.
	\begin{align*}
		\mathcal{P}\left( d \right) = \left\{ \sum_{i = 0}^{d - 1} p_i z^i \, \mid \, p_i \in \C, \forall i \in \left[ d - 1 \right] \right\}
	\end{align*}
	Suppose $p \in \mathcal{P}\left( d \right)$. Then the \emph{Fourier transform} of $p$ is the vector
	\begin{align*}
		\widehat p = \begin{pmatrix}
			             p(1)             \\
			             p(e^{2 \pi i/d}) \\
			             p(e^{4 \pi i/d}) \\
			             \vdots           \\
			             p(e^{2 \pi i \left( d - 1\right)/d})
		             \end{pmatrix}
	\end{align*}
	obtained by evaluating $p$ at each of the $d$-th roots of unity.
\end{dfn}

\begin{thm}[\textbf{Parseval's Theorem}]
	Let $f$ be a function over $\zmodn$, let $\hat{f}$ is its Fourier transform, an $N$-dimensional complex vector with entries $\hat{f}_k$. Then
	\[ \sum_{k\in N} \abs{f(k)}^2 = \frac{1}{N} \sum_{k\in N} \abs{\hat{f}_k}^2 . \]
\end{thm}

We now re-state and prove the main theorem of this section.

\noindent\textbf{\Cref{thm:small-fourier-implies-orn}.}  {\em
Given a shift schedule $\bm{\sigma} = s_0,s_1,\ldots,s_{T-1}$ on $G = \zmodn$, and given parameters $\Lambda$ and $h$ for which $T\geq \Lambda h$.

If $\Lambda h$ evenly divides $T$, then let $\mathcal{A} = \{\ell \Lambda h : \ell \in \mathbb{Z}_{\geq 0}\}$. Otherwise, let $\mathcal{A} = \{0,\ldots,T-1\}$.
Then for each timestep $t \in \mathcal{A},$
let $p_t$ and $q_t$ be the polynomials\footnote{Note that $p_t$ and $q_t$ are generating polynomials of forward and backwards routing distributions which delay flow until timestep $t$, then use $h$ spraying hops chosen uniformly at random, each from consecutive blocks of $\Lambda$ timesteps.}
\[ p_t(z) = \prod_{j=1}^{h} p_{t,j}(z) , \hspace{5mm} q_t(z) = \prod_{j=h+1}^{2h} q_{t,j}(z) \]
where
\[ p_{t,j}(z) = \frac{1}{\Lambda} \left( \sum_{k= t + (j-1)\Lambda}^{t + j\Lambda-1} z^{s_{k}} \right) , \hspace{5mm} q_{t,j}(z) = \frac{1}{\Lambda} \left( \sum_{k=t+(j-1)\Lambda}^{t+j\Lambda-1} z^{N-s_{k}} \right) \]
Define $\hat{p_t}^\star$ as the Fourier transform of $p$ with first coordinate set to 0, and $\hat{q_t}^\star$ similarly for $q$.

Now let $\eps>0$ be a constant.
If $\norm{\hat{p_t}^\star}_2 \leq \frac{\eps}{2}$ and $\norm{\hat{q_t}^\star}_2 \leq \frac{\eps}{2}$ for all $t\in\mathcal{A}$, then there exists an oblivious routing protocol over $\bm{\sigma}$ which guarantees throughput $\frac{1}{2h}(1 - \eps)$ and maximum latency no more than $2(h+1)\Lambda$.
}

\begin{proof}
	Consider our routing distribution $\mathcal{D}$ for connection schedule $\bm{\sigma}$, defined in \Cref{def:h-hop-spraying},
	which first delays flow until a timestep $t\in\mathcal{A}$, then equally weights all paths which take exactly one hop over the next consecutive subsequences each of length $\Lambda$.
	Fixing source node 0, and timestep $t$, this forms a distribution on destination nodes $b$ with generating polynomial
	\[ p_t(z) = \prod_{j=1}^{h} p_{t,j}(z), \hspace{2mm}\mbox{ {\normalfont where} }\hspace{2mm} p_{t,j}(z) = \frac{1}{\Lambda} \left( \sum_{k=t+(j-1)\Lambda}^{t+j\Lambda-1} z^{s_{k}} \right) . \]

	This has a Fourier transform that can be decomposed into the element-wise product
	\[ \hat{p_t}(z) = \bigodot_{j=1}^{h} \hat{p_{t,j}}(z) , \]
	where
	\begin{align*}
		\hat{p_{t,j}} = \frac{1}{\Lambda}
		\begin{pmatrix}
			1                                                                          \\
			\vspace{2mm} 
			\sum_{k=t+(j-1)\Lambda}^{t+j\Lambda-1} \mbox{ } e^{\frac{2\pi i}{N} s_{k}} \\
			\sum_{k=t+(j-1)\Lambda}^{t+j\Lambda-1} \mbox{ } e^{\frac{4\pi i}{N} s_{k}} \\
			\vdots                                                                     \\
			\sum_{k=t+(j-1)\Lambda}^{t+j\Lambda-1} \mbox{ } e^{\frac{(N-1)\pi i}{N} s_{k}}
		\end{pmatrix}
	\end{align*}

	By assumption, we know that $\norm{\hat{p_t}^\star}_2 \leq \frac{\eps}{2}$.
	Since $\bm{u}$ is the uniform distribution, $\hat{u} = \bm{e}_1$, the first standard basis vector $(1,0,0,\ldots,0)^T$.
	Then $\hat{p_t}^\star = \hat{p_t}-\hat{u}$, because the first coordinate of $\hat{p_{t}}$ equals $1$.
	\[ \hat{p_{t}}[0] = \bigodot_{j=1}^{h} \frac{1}{\Lambda} \left( \sum_{k=t+(j-1)\Lambda}^{t+j\Lambda-1} e^{0} \right) = \bigodot_{j=1}^{h} \frac{1}{\Lambda}(\Lambda) = 1 \]
	Therefore,
	\[ \norm{\hat{p_t}^\star}_2 \leq \frac{\eps}{2} \implies \norm{\hat{p_t}-\hat{u}}_2 \leq \frac{\eps}{2} \]

	Use Parseval's Theorem to relate $\norm{\hat{p_t}-\hat{u}}_2$ to $\norm{p_t-u}_2$, which then bounds $\norm{p_t - u}_1$ as desired.
	For all $t\in\mathcal{A}$,
	\begin{align*}
		\norm{\hat{p_t}-\hat{u}}_2 \leq \frac{\eps}{2} & \implies \norm{p_t - u}_2 \leq \frac{\eps}{2} N^{-1/2} \\
		                                               & \implies \norm{p_t - u}_1 \leq \frac{\eps}{2} .
	\end{align*}

	Similarly, consider our backwards routing distribution $\mathcal{E}$ for connection schedule $\bm{\sigma}$, which delays flow backwards until some $t\in\mathcal{A}$ for some non-negative integer $\ell$, and then takes one uniformly random hop backwards from each of $h$ consecutive sub-sequences $\bm{\sigma}_i$ of length $\Lambda$.
	A similar calculation shows that $\norm{q_t - u}_1 \leq \frac{\eps}{2}$ for all $t\in\mathcal{A}$.
	Apply \Cref{lem:vlb-w-leakage} to obtain the final result.

\end{proof}

\subsection{A Test for Universality}

	The connections between universal shift connection schedules and Fourier analysis are not unique to the three constructions presented in the following sections of this paper.
	Given any deterministic shift connection schedule for a fixed network size $N$, one could numerically test whether such a connection schedule is universal by computing the Fourier coefficient vectors of the generating polynomials
	directly.\footnote{Of course, this test can only conclude whether the $h$-hop packet spraying-inspired VLB with Leakage routing protocols are optimal on this connection schedule for each value $h$ or not.
			If they are not optimal, that does not imply that no such optimal routing protocols exist for each $h$.
	This can be done in polynomial time in the size of the network, $N$.
	}
	Such a test could even be performed on non-shift connection schedules.
	However as we discuss in \Cref{sec:conclusion}.2, bounding the closeness to uniformity by computing the Fourier coefficients of a generating polynomial does not work for non-shift connection schedules, leading us to need a different approach.

%
%

\section{Random Shift Schedule}\label{sec:rand-shift-sched}

In this section, we prove that when using $h$-hop packet spraying and VLB with Leakage, then for a suitable choice of phase length $\Lambda_h$, a uniformly random shift connection schedule is a universal connection schedule with constant probability.
Applying \Cref{thm:small-fourier-implies-orn}, we must only prove that for each choice of $h\in\{1,\ldots,\log_2 N\}$ and $\Lambda_h$, the random Fourier coefficient vector associated with our random shift connection schedule, $h$, and $\Lambda_h$ has 2-norm that is bounded by $\frac{\eps}{2}$ with high probability.
We will prove this bound for each individual $h$, and then apply a union bound over all $h$.
That is exactly the focus of the theorem stated below.

\begin{thm}\label{thm:random-shift-works}
	Given a constant $\eps>0$, hop parameter $h$, and a network size $N$,
  let $\Lambda \geq 4 \left(\frac{\eps}{2}\right)^{-2/h} N^{1/h} \ln(hN^3)$. Let $\bm{\sigma} = s_0, \ldots, s_{T - 1}$ be a sequence of shifts each chosen uniformly at random from $\zmodn$, and assume that both $T\leq \frac{N^2}{8\log_2 N}$ and $T\geq h\Lambda$.

	Then for each timestep $t \in \{0,\ldots,T-1\}$, let $p_t$ be the polynomial
	\[ p_t(z) = \prod_{j=1}^{h} p_{t,j}(z), \quad \text{where } p_{t,j}(z) = \frac{1}{\Lambda} \left( \sum_{k=t+(j-1)\Lambda}^{t+j\Lambda-1} z^{s_{k}} \right). \]

	Then 
	\[ \Pr\Big[ \norm{\hat{p}_t^\star}_2 \leq \frac{\eps}{2} \mbox{ for every } t\in\{0,\ldots,T-1\} \Big] \geq 1 - \frac{1}{2\log_2 N} \]
	where $\hat{p}_t^\star$ is the Fourier coefficient vector of $p_t$ with first entry set to zero.

\end{thm}

\begin{proof}
	For ease of notation, assume that $\xi = e^{2\pi i/N}$, and note that $p_{t,j}(\xi^a)$ is the $a$th coordinate of the Fourier coefficient vector of $p_{t,j}$.
	First suppose for all $j\in[h]$, $t\in[T]$, $a\in[N-1]$, that
	\[ \abs{\Re(p_{t,j}(\xi^a))} \leq \beta \mbox{ and } \abs{\Im(p_{t,j}(\xi^a))} \leq \beta  \]
	for some parameter $\beta$. Then it would be the case that
	\begin{align*}
		\norm{\hat{p_t}^\star}_2^2 & \le \sum_{a = 1}^{N - 1} \left(\prod_{j = 1}^h p_{t,j}(\xi^a) \right)^2 \\
		                           & \le N \cdot (\sqrt{2}\beta)^{2h} 	
	\end{align*}

	If we desire $\norm{\hat{p}_t^\star}_2 \leq \frac{\eps}{2}$, then setting $\beta = \frac{(\eps/2)^{1/h}}{\sqrt{2}N^{1/2h}}$ suffices.
	Therefore, in order to bound $\Pr\left[\norm{\hat{p}_t^\star}_2 > \frac{\eps}{2} \right]$, it suffices to bound
	\[ \Pr \left[ \abs{\Re(p_{t,j}(\xi^a))} > \frac{(\eps/2)^{1/h}}{\sqrt{2}N^{1/2h}} \right] \mbox{ and } \Pr \left[ \abs{\Im(p_{t,j}(\xi^a))} > \frac{(\eps/2)^{1/h}}{\sqrt{2}N^{1/2h}} \right] , \]
	and then union bound over all possible values of $t,j$, and $a$, and over the real and imaginary parts.

	In order to bound these probabilities, we examine their expected values.
	By linearity of expectation, 
		\begin{align*}
		\E\left[p_{t,j}(\xi^{a})\right] & = \E\left[ \frac{1}{\Lambda} \sum_{k} \exp\left( \frac{2\pi i}{N} a s_k \right) \right]                   \\
		                                & = \frac{1}{\Lambda} \sum_{k} \E_{s_k \in[N]}\left[ \exp\left( \frac{2\pi i}{N}a s_k \right) \right] = 0 .
	\end{align*}
	Therefore $\E[\Re(p_{t,j}(\xi^{a}))] = \E[\Im(p_{t,j}(\xi^{a}))] = 0$ as well.

	As each $s_k$ is independently sampled, we can view $\Re(p_{t,j}(\xi^{a})) = \sum_k \frac{1}{\Lambda} \Re\left(\exp\left( \frac{2\pi i}{N}a s_k \right)\right)$ and 
        $\Im(p_{t,j}(\xi^{a})) = \sum_k \frac{1}{\Lambda} \Im\left(\exp\left( \frac{2\pi i}{N}a s_k \right)\right)$ as sums of $\Lambda$ random variables, each supported on $[\frac{-1}{\Lambda},\frac{1}{\Lambda}]$.
	We then apply Hoeffding's inequality.

	\begin{thm}[\textbf{Hoeffding's Inequality}]\label{thm:hoeffding}
		Let $X_1,\ldots,X_n$ be independent random variables such that $a_j\leq X_j\leq b_j$ almost surely.
		Then
		\[ \Pr\left[\abs{\sum_{j=1}^{n} X_j - \E\Big[ \sum_{j=1}^{n} X_j \Big]} \geq t \right] \leq \exp\left( \frac{-2t^2}{\sum_{j=1}^{n} (b_j-a_j)^2} \right) . \]
	\end{thm}

	Using Hoeffding's inequality, we find that
	\begin{align*}
		\Pr \left[ \abs{\Re(p_{t,j}(\xi^a))} > \frac{(\eps/2)^{1/h}}{\sqrt{2}N^{1/2h}} \right] \leq 2\exp\left( \frac{-(\eps/2)^{2/h}\Lambda}{4N^{1/h}} \right) 
	\end{align*}
    and the same tail bound holds for imaginary parts.
	Union bounding over all possible values of $t,j$, and $a$, and over the real and imaginary parts, this shows that
	\begin{align*}
		\Pr\left[ \exists T \in [T], j \in [h] \mbox{ s.t. } \norm{\hat{p}_{t,j}^\star}_2 > \frac{\eps}{2} \right] \leq 2TNh\cdot 2\exp\left( \frac{-(\eps/2)^{2/h}\Lambda}{4N^{1/h}} \right)
	\end{align*}
  Assuming $T\leq \frac{N^2}{8\log_2 N}$ and $\Lambda$ is the smallest integer for which $\Lambda\geq 4\left(\frac{\eps}{2}\right)^{-2/h} N^{1/h} \ln\left(hN^3\right)$, this implies that
	\begin{align*}
		\Pr\left[ \exists T \in [T], j \in [h] \mbox{ s.t. } \norm{\hat{p}_{t,j}^\star}_2 > \frac{\eps}{2} \right] \leq \frac{1}{2\log_2 N} .
	\end{align*}

\end{proof}

\begin{cor}\label{cor:multiple-h-small-fourier}
	Given a constant $\eps>0$ and a network size $N$,
	let $T = 4\left(\frac{\eps}{2}\right)^{2} N \ln\left(N^3\right)$, and let $\bm{\sigma} = s_0, \ldots, s_{T - 1}$ be a sequence of shifts chosen uniformly at random from $\zmodn$.

  For each $t\in\{0,\ldots,T-1\}$ and $h\in[\log_2 N]$, fix $\Lambda_h = 4\left(\frac{\eps}{2}\right)^{-2/h} N^{1/h} \ln\left(hN^3\right)$ and let $p_{h,t}(z)$ be the polynomial
	\[ p_{h,t}(z) = \prod_{j=1}^{h} p_{h,t,j}(z), \quad \text{where } p_{h,t,j}(z) = \frac{1}{\Lambda_h} \left( \sum_{k=t+(j-1)\Lambda_h}^{t+j\Lambda_h-1} z^{s_{k}} \right). \]
	Then
	\[ \Pr\Big[ \forall h \in [\log{N}] \mbox{ and } t\in\{0,\ldots,T-1\}, \;\; \|\hat{p}_{h,t}^\star\|_2 \leq \frac{\eps}{2} \Big] > \frac{1}{2} . \]
\end{cor}

\begin{proof}
	For $h \in [\log N]$, define $E_h$ to be the event that $\norm{\hat{p}_h^\star}_2 > \frac{\eps}{2}$.

  By \Cref{thm:random-shift-works}, we know that for $\Lambda_h = 4\left(\frac{\eps}{2}\right)^{-2/h} N^{1/h} \ln\left(hN^3\right)$, we have
	\[ \Pr[E_h] \leq \frac{1}{2\log_2 N} \]

	Taking the union bound over all $h \in [\log_2 N]$, we get $\Pr[\bigcup_{h = 1}^{\log N} E_h] \leq \frac{1}{2}$.

\end{proof}

\begin{cor}[\textbf{\Cref{thm:main-thm}.1}]
	For sufficiently large $N$, there exists a randomized algorithm to compute a shift connection schedule $\bm{\sigma}$ on $N$ nodes which succeeds with probability at least $\frac{1}{2}$ and is efficiently certifiable.
	Additionally for each $h\in \{ 1,\ldots,\log_2N \}$, there also exists an oblivious routing protocol on $\bm{\sigma}$ which guarantees throughput $\frac{1}{2h}(1-\eps)$ and achieves maximum latency $\bigotilde_\eps\left(L^\star(h,N)\right)$.
\end{cor}

\begin{proof}
  The randomized algorithm samples a uniformly random shift schedule
    $\bm{\sigma} = s_0, s_1, \ldots, s_{T-1}$. The certification procedure
    works as follows: for every $h \leq \log_2(N)$, compute the vectors
    $\hat{p}_t^{\star}$ defined in \Cref{thm:small-fourier-implies-orn}
    for every $t < T$, and verify that each of the vectors has 2-norm
    at most $\frac{\eps}{2}$. By \Cref{cor:multiple-h-small-fourier}
    the random shift schedule passes this test with probability greater
    than $\frac12$, and by \Cref{thm:small-fourier-implies-orn}, when
    the test passes the routing protocol defined by VLB with Leakage
    will satisfy the claimed throughput guarantee.
  
\end{proof}

\section{Convolution Schedule}\label{sec:convolution-sched}

We devote this section to proving the below theorem. From this, \Cref{thm:main-thm}.2 becomes a corollary.

\begin{thm} \label{thm:main-convolve-thm}
	Given any set of constants $\{h_1,h_2,\ldots,h_\ell\}$ and any constant $\eps>0$, then for sufficiently large network sizes $N$ there exists a randomized algorithm to compute a shift connection schedule $\bm{\sigma}$ on $N$ nodes which succeeds with probability at least $\frac{1}{2}$ and is efficiently certifiable.
	Additionally for each $j\in\{1,\ldots,\ell\}$, there also exists an oblivious routing protocol on $\bm{\sigma}$ which guarantees throughput $\frac{1}{2h_j}(1-\eps)$ and achieves maximum latency $\bigo(L^\star(h_j,N))$, where the constant in the $\bigo(\cdot)$ depends only on $\eps$ and on the least common multiple of $h_1,\ldots,h_{\ell}$.
\end{thm}

In order to prove this theorem, we describe how to build a shift connection schedule with desirable Fourier properties by convolving smaller shift schedules together, each of which has desirable Fourier properties that will allow us to use \Cref{thm:small-fourier-implies-orn}.
Building up schedules in this way will only guarantee that a constant number of throughput-latency tradeoffs can achieved on the same connection schedule.
However, we will derive a latency bound that is provably optimal up to a constant factor \cite{optimal-orns}, unlike in \Cref{sec:rand-shift-sched}.
This is because instead of requiring a large union bound, we will be able to make use of the following inequality 
pertaining to martingales in Banach spaces.

\begin{prop}[\textbf{Prop 4.35 in \cite{pisier_2016}}]\label{thm:pisier}
	Let $\left( \Omega,A,P \right)$ be a probability space.
	Let $Y_n$ be a sequence of independent mean-zero random variables in Banach space $B$.
	Assume that $B$ is of type $q$, and that the series $\sum_n Y_n$ is a.s. convergent.
	Then for $0 < \alpha < \infty$, we have
	\[  \E\left[ \norm{\sum_n Y_n}^\alpha \right] \leq \beta\E\left[\left( \sum_n\norm{Y_n}^q \right)^{\alpha/q} \right] \]
	where $\beta$ is a constant that depends only on $\alpha$ and $B$.
\end{prop}

We will now state and prove the following Lemma, which shows the existence of Base Schedules with desirable Fourier properties.

\begin{lem} \label{lem:applied-pisier-bound}
	Given a constant $\eps>0$, a positive integer $H$, let $\beta$ be a constant that depends on $H$.
	Then given a network size $N$ and 	uniformly random shifts $s_1,\ldots,s_\Lambda$ drawn independently from $\zmodn$ with $\Lambda \geq \frac{2\beta (2N)^{1/H}}{\eps^{1/H}}$, let polynomial $p$ be
	\[ p(z) = \frac{1}{\Lambda} \left( \sum_{k=1}^{\Lambda} z^{s_{k}} \right)  \]
	and let $\hat{p}^\star$ be the Fourier transform vector of $p$ with first index set to $0$.
	Then
	\[ \E\left[ \norm{\hat{p}^\star}_{2H} \right] \leq \frac{1}{2}\left(\frac{\eps}{2}\right)^{1/H} . \]
\end{lem}

\begin{proof}
	Because the shifts $s_k$ are independent and uniformly distributed, $\hat{p}^\star$ is the average of i.i.d. random vectors, each of which is the Fourier transform of a single monomial, $z^{s_k}$.
	By padding these $N$-dimensional vectors with an infinite sequence of 0's, we will consider them as elements $\hat{p}_k^{\star}$ of the infinite-dimensional complex Banach space $\ell^{2H}$.
	The random vectors $\hat{p}_k^\star$ have mean $\bm{0}$ and are uniformly distributed over the set of $N$ possible Fourier coefficient vectors.

	Since $\E\left[ X \right] \leq \sqrt{\E[X^2]}$ for any random variable $X$, let us now try to bound $\E\left[ \norm{\hat{p}^\star}_{2H}^2 \right] = \E\left[ \norm{\sum_{k=1}^{\Lambda} \frac{1}{\Lambda} \hat{p}_k^\star}_{2H}^2\right]$.
	Apply \Cref{thm:pisier} by interpreting $B$ as $\ell^{2H}$, the random variables $Y_k$ as our random vectors $\frac{1}{\Lambda} \hat{p}_k^\star$, and $\alpha=2$.
	Note that our Banach space has type $2$ (and co-type $2H$).
	We can then derive
	\begin{align*}
		\E\left[ \norm{\sum_{k=1}^{\Lambda} \frac{1}{\Lambda} \hat{p}_k^\star}_{2H}^2\right] \leq \beta \E\left[ \sum_{k=1}^{\Lambda} \norm{\frac{1}{\Lambda} \hat{p}_k^\star}_{2H}^{2} \right]
	\end{align*}
	for a constant $\beta$ that depends only the Banach space $\ell^{2H}$ and the constant $\alpha=2$.
	That is, $\beta$ is only a function of $H$.

	We may decompose this sum by linearity of expectation.
	\begin{align*}
		\beta \E \left[ \sum_{k=1}^{\Lambda} \norm{\frac{1}{\Lambda} \hat{p}_k^\star}_{2H}^{2} \right]
		 & = \beta \sum_{k=1}^{\Lambda} \E\left[ \norm{\frac{1}{\Lambda} \hat{p}_k^\star}_{2H}^{2} \right]                                                                            \\
		 & = \beta \sum_{k=1}^{\Lambda} \sum_{j=1}^{N} \frac{1}{N} \left( \sum_{\gamma=1}^{N-1} \abs{\frac{1}{\Lambda} \exp\left(\frac{2\pi i}{N}\gamma j\right)}^{2H} \right)^{2/2H} \\
		 & = \beta \sum_{k=1}^{\Lambda} \frac{1}{N} \sum_{j=1}^{N} \left( (N-1) \left(\frac{1}{\Lambda}\right)^{2H} \right)^{1/H}                                                    \\
		& = \beta \frac{(N-1)^{1/H}}{\Lambda} <
                   \frac{1}{2}\left(\frac{\eps}{2}\right)^{1/H}
	\end{align*}
        by our assumption that $\Lambda \geq \frac{2\beta (2N)^{1/H}}{\eps^{1/H}}$.
\end{proof}

\noindent There is one more theorem to state before we can prove \Cref{thm:main-convolve-thm}. 

\begin{thm}[\textbf{H\"{o}lder's Inequality}]\label{thm:holders-ineq}
	Let $\bm{v_1},\ldots,\bm{v_k}$ be real or complex-valued vectors, and suppose $r$ and $q_1,\ldots,q_k$ are positive numbers such that $\frac{1}{r} = \sum_{j=1}^{k} \frac{1}{q_j}$.
	Then
	\[ \norm{\bigodot_{j=1}^{k} \bm{v_j} }_r \leq \prod_{j=1}^{k} \norm{\bm{v_j}}_{q_j} . \]
\end{thm}

We now turn to proving \Cref{thm:main-convolve-thm}.

\begin{proof}
	Recall that $\eps>0$, a network size $N$, and a set of hop-counts $\{h_1,\ldots,h_\ell\}$ are given.
	Set $H$ equal to the least common multiple of $\{h_1,\ldots,h_\ell\}$. 
	Let $\Lambda = \left\lceil \frac{2\beta(2N)^{1/H}}{\eps^{1/H}} \right\rceil.$
	Then using \Cref{lem:applied-pisier-bound}, we can sample a shift connection schedule with the property
	\[ \E\left[ \norm{\hat{p}^\star}_{2H} \right] \leq \frac{1}{2}\left(\frac{\eps}{2}\right)^{1/H} . \]
	Then by Markov's Inequality, with probability at least $\frac{1}{2}$, we sample a schedule with 
	\[ \norm{\hat{p}^\star}_{2H} \leq 2\cdot \E\left[ \norm{\hat{p}^\star}_{2H} \right] = \left(\frac{\eps}{2}\right)^{1/H} . \]
	Call this a {\em Base Schedule}, and find $H$ Base Schedules\footnote{We may use $S_1=\hdots=S_H$, which only requires sampling a single Base Schedule with the desired properties.} $S_1,\hdots,S_H$, each of length $\Lambda$, and each with $\norm{\hat{p_j}^\star}_{2H} \leq \left(\frac{\eps}{2}\right)^{1/H}$.
	Note that certifying that any given Base Schedule
          obeys the inequality
          $\norm{\hat{p}^\star}_{2H} \leq \left(\eps/2\right)^{1/H}$
          requires only
        polynomial time in the number of nodes, $N$.

	Denote Base Schedule $S_j = \{ S_j(0),S_j(1),\ldots,S_j(\Lambda-1) \}$.
	Convolve these Base Schedules together to construct a single shift connection schedule with period $T = \Lambda^H$.
	To do so, represent timesteps as $H$-tuples, $t=(t_1,\ldots,t_H)$, where $t = t_1 + \Lambda t_2 + \ldots + \Lambda^{H-1} t_H$.
	Our shift connection schedule will be defined by the following function:
	at timestep $t=(t_1,\ldots,t_H)$, take shift $s_t=\sum_{j=1}^{H} S_j(t_j)$.

	In addition, for every divisor $h$ of $H$, we will define an oblivious routing protocol on our shift schedule which achieves maximum latency $\bigo\left( h N^{1/h} \right)$ and guarantees throughput $\frac{1}{2h}(1-\eps)$.
	We begin with defining the oblivious routing protocol for some special cases of $h$, and then show the general case.

	$\bm{(h=1)}$.
	In this routing protocol, delay flow until the start of a period.
	Our forward routing distribution is defined by taking exactly one uniformly random hop during a full period of the schedule.
	Our backwards routing distribution is defined by taking exactly one uniformly random hop backwards during a full period of the schedule.
	Use VLB with Leakage to create an oblivious routing protocol with maximum latency no more than $3N = \bigo(N)$.

	To show this routing protocol guarantees throughput $\frac{1}{2}(1-\eps)$, we will apply \Cref{thm:small-fourier-implies-orn}.
	Thus, we need to examine the generating polynomial of the routing distribution, which we denote by $\mathcal{D}$.
          The generating polynomial $p$ of $\mathcal{D}$ satisfies
          \begin{align*}
            p(z) & = \frac{1}{T} \sum_{t = 0}^{T} z^{s_t} =
            \frac{1}{\Lambda^H} \sum_{(t_1,\ldots,t_H) \in [\Lambda]^H}
            z^{\sum_{j=1}^H S_j(t_j)} 
            = \frac{1}{\Lambda^H} \sum_{(t_1,\ldots,t_H) \in [\Lambda]^H}
            \prod_{j=1}^H z^{S_j(t_j)} \\
            & = \prod_{j=1}^H \left( \frac1\Lambda \sum_{\gamma=0}^{\Lambda-1}
            z^{S_j(\gamma)} \right) 
            = \prod_{j=1}^H p_j(z) \qquad \mbox{{\normalfont where}}\hspace{2mm} p_j(z) = \frac{1}{\Lambda}\left( \sum_{\gamma=0}^{\Lambda-1} z^{S_j(\gamma)} \right) .
          \end{align*}

	Therefore,
	\[ \hat{p}^\star = \bigodot_{j=1}^{H} \hat{p}_j^\star . \]

	In order to apply \Cref{thm:small-fourier-implies-orn}, we must bound the $2$-norm of $\hat{p}^\star$.
	We apply H\"{o}lder's Inequality.
	\begin{align*}
		\norm{\hat{p}^\star}_2 & \leq \prod_{j=1}^{H} \norm{\hat{p}_j^\star}_{2H} \\
		                       & \leq \left( \left(\frac{\eps}{2}\right)^{1/H} \right)^H = \frac{\eps}{2}
	\end{align*}
	A similar derivation holds for the backwards routing distribution, and thus, applying \Cref{thm:small-fourier-implies-orn} gives our desired guaranteed throughput bound.

	$\bm{(h>1)}$.
    To define the routing protocol we require that $H$ is divisible by $h$.
        Let $g = H/h$.
	In this routing protocol, delay flow until a timestep $t$ with $t_1,\ldots,t_{g} = 0$, and some arbitrary $t_{g + 1},\ldots t_{H}$.
	Our forward routing distribution is defined by taking $h$ total hops, exactly one during each of the next $h$ blocks of $\Lambda^{g}$ timesteps,
        \[
          [t,t+\Lambda^{g}-1], \; \;
          [t+\Lambda^g, t+2\Lambda^g - 1], \; \;
          \ldots, \; [t+(h-1)\Lambda^{g},t+h\Lambda^{g}-1] .
          \]
	Our backwards routing distribution is defined similarly, by taking exactly one backwards hop during each of $[t+h\Lambda^{g},t+(h+1)\Lambda^{g}-1]$ through $[t+(2h-1)\Lambda^{g},t+2h\Lambda^{g}-1]$.
	Use VLB with Leakage to create an oblivious routing protocol with maximum latency no more than
     $\bigo(h \Lambda^g)$.
      By our choice of $\Lambda =
      \left\lceil \frac{2\beta(2N)^{1/H}}{\eps^{1/H}} \right\rceil,$
      we have
      $\Lambda^H = \bigo(N)$ where the (large) constant inside the $\bigo$
      depends on $\eps$ and $H$ but not on $N$. Hence, 
      \[ \bigo(h \Lambda^g) = \bigo(h \Lambda^{H/h}) = \bigo(h N^{1/h})
      = \bigo(L^*(h,N)).\]
        
	To show that this routing protocol guarantees throughput $\frac{1}{2h}(1-\eps)$, we will apply \Cref{thm:small-fourier-implies-orn}.
     According to the proof of that theorem, the generating
          polynomial $p_t(z)$ of the forward routing distribution factorizes
      as $p_t(z) = \prod_{j=1}^h p_{t,j}(z)$ where:
      \begin{align*}
        p_{t,j}(z) & = \frac{1}{\Lambda^{g}} \sum_{t' = t + (j-1) \Lambda^{g}}^{t + j \Lambda^{g} - 1} z^{s_t} .
      \end{align*}
      To simplify $p_{t,j}(z)$, recall $t$ is divisible
      by $\Lambda^g$. Write $t + (j-1) \Lambda^{g}$
      in base $\Lambda$ as
      \begin{equation} \label{eq:base-lambda}
      t + (j-1) \Lambda^g = \Lambda^g t_{g+1} + \Lambda^{g+1} t_{g+2} + \cdots + \Lambda^{H-1} t_{H-1}.
      \end{equation}
      Then
      \begin{align}
      p_{t,j}(z) &= \frac{1}{\Lambda^{g}} \sum_{(t_1,\ldots,t_g) \in [\Lambda]^g}
        z^{\sum_{k=1}^H S_k(t_k)} 
        = \frac{1}{\Lambda^g} \sum_{(t_1,\ldots,t_g) \in [\Lambda]^g}
        \prod_{k=1}^H z^{S_k(t_k)} \nonumber \\
        & = \prod_{k=1}^g \left( \frac1\Lambda \sum_{\gamma=0}^{\Lambda-1}
        z^{S_k(\gamma)} \right) \cdot \prod_{k=g+1}^{H} z^{S_k(t_k)} \nonumber \\
        & = \left( \prod_{k=1}^g p_{[k]}(z) \right) \prod_{k=g+1}^{H} z^{S_k(t_k)} \qquad \mbox{{\normalfont where}}\hspace{2mm} p_{[k]}(z) = \frac{1}{\Lambda}\left( \sum_{\gamma=1}^{\Lambda} z^{S_k(\gamma)} \right) .
        \label{eq:factorization-ptj}
      \end{align}
        
	In order to apply \Cref{thm:small-fourier-implies-orn}, we must bound the 2-norm of $\hat{p}_t^\star$. For $j \in \{1,2,\ldots,h\}$ let
        $f_{t,j}(z) = \prod_{k=g+1}^{H} z^{S_k(t_k)}$ denote
        the monomial appearing in the factorization~\eqref{eq:factorization-ptj}, where the digits $t_{g+1},\ldots,t_H$ appearing in the expression for $f_{t,j}(z)$ implicitly depend on $t$ and $j$ via Equation~\eqref{eq:base-lambda}.
        Note that, since $f_{t,j}(z)$ is a monomial in $z$,
        every Fourier coefficient $\hat{f}_{t,j}$ is a complex number
        of norm 1. Hence, for any other vector $v$,
        \[
        \norm{v \odot \hat{f}_{t,j}}_{2} = \norm{v}_{2} .
        \]
        
    Now, using H\"{o}lder's Inequality,
    \begin{align*}
      \norm{\hat{p}_t^\star}_2 & =
      \norm{\bigodot_{j=1}^{h} \hat{p}_{t,j}^\star}
      = \norm{\bigodot_{j=1}^h \left( \hat{f}_{t,j}^\star \odot
        \bigodot_{k=1}^{g} \hat{p}_{[k]}^\star \right)}_{2}
      = \norm{\bigodot_{j=1}^h 
        \bigodot_{k=1}^{g} \hat{p}_{[k]}^\star}_{2} \\
      & \leq \prod_{j=1}^{h} \prod_{k=1}^{g} \norm{\hat{p}_{[k]}^\star}_{2gh}
      = \prod_{j=1}^h \prod_{k=1}^{g} \norm{\hat{p}_{[k]}^\star}_{2H}
      \leq \left( \left( \frac{\eps}{2} \right)^{1/H} \right)^{gh} = \frac{\eps}{2}.
    \end{align*}
        
	A similar derivation holds for the backwards routing distribution, and thus, applying \Cref{thm:small-fourier-implies-orn} gives our desired guaranteed throughput bound

	\textbf{In conclusion.} 
	Sampling a Base Schedule with the desired Fourier properties occurs with probability at least $\frac{1}{2}$.
	We convolve our Base Schedule with itself to build a larger shift connection schedule.
	On this shift connection schedule, for any integer $h$ which divides $H$, there is an oblivious routing protocol which achieves maximum latency no more than $\bigo(hN^{1/h}) = \bigo(L^\star(h,N))$ and guarantees throughput $\frac{1}{2h}(1-\eps)$.

\end{proof}

\begin{cor}[\textbf{\Cref{thm:main-thm}.2}]
	Given a constant $\eps>0$, for sufficiently large network sizes $N$ there exists a randomized algorithm to compute a shift connection schedule $\bm{\sigma}$ on $N$ nodes which succeeds with probability at least $\frac{1}{2}$ and is efficiently certifiable.
	Additionally for each $h\in\{ 1,2,\ldots,\log_2 N \}$ there exists an oblivious routing protocol on that shift connection schedule which guarantees throughput $\frac{1}{2h}-\eps$.
	and achieves maximum latency $\bigo\left(L^\star(h,N)\right)$.
\end{cor}

\begin{proof}
	Note that when $\frac{1}{2h}\leq \eps$, then $\frac{1}{2h}-\eps\leq 0$.
	Any oblivious routing protocol, even one which does not route any flow, guarantees throughput $0$.
	We therefore only need to handle the case when $\frac{1}{2h}>\eps$.

	Let $\{h_1,\ldots,h_l\}$ be the set of positive integers with $h<\frac{1}{2\eps}$.
	As $\eps$ is a constant, this set also has constant size.
	Apply \Cref{thm:main-convolve-thm}.
	The routing protocols returned by \Cref{thm:main-convolve-thm} guarantee throughput $\frac{1}{2h}(1-\eps)$. Since $\frac{1}{2h}\leq \frac{1}{2}$, we have $\frac{1}{2h}(1-\eps)\geq \frac{1}{2h}-\eps$.

\end{proof}

\newcommand{\calW}{\mathcal{W}}

\section{Derandomized Construction}\label{sec:derandomized}

In this section, we aim to provide a deterministic, computationally efficient construction of a universal connection schedule.
Previous work~\cite{optimal-orns} conjectured that the following deterministic connection schedule inspired by \cite{tremel} would be universal:
  Let $\mathbb{F}$ denote the finite field with $N$ elements, and let $x$ denote a primitive root in $\mathbb{F}$.
  Define the sequence of permutations $\pi_0, \pi_1, \ldots$ by specifying that $\pi_k(i) = i + x^k$ for all $i \in \mathbb{F}, \, k \in \mathbb{N}$.
  Notably, the conjectured design is a shift connection schedule, and thus can be tested for individual values of $N$ using \Cref{thm:small-fourier-implies-orn} by computing the Fourier coefficient vectors of the generating polynomials directly.
  However, we make no further progress toward proving this design is universal for infinitely many network sizes $N$.

Instead, we describe a deterministic construction of a universal shift connection schedule that is motivated directly by derandomizing the techniques we used in \Cref{sec:small-fourier-implies-orn,sec:rand-shift-sched}.
Like in all previous sections, we will build a connection schedule, show that this connection schedule has desirable Fourier properties, and then use \Cref{thm:small-fourier-implies-orn} to show the existence of oblivious routing protocols that guarantee any throughput value up to a multiplicative error of $(1-\eps)$.

At a high level, our strategy is the following recursive approach.
We assume $N$ is a power of 2 and construct a universal shift connection schedule with period length $T = N$ in which the shifts $s_0, s_1, \ldots, s_{N-1}$ constitute a permutation of the elements of $\zmodn$. 
The permutation is constructed by using a recursive partitioning process to define a total ordering of the elements of $\zmodn$.
Take the set of all possible shifts to start.
Divide this set into 2 equal-size pieces $S_{\text{left}}$ and $S_{\text{right}}$, in such a way that the ``Fourier 2-norm'' does not increase by very much.
Our total ordering will place all of $S_{\text{left}}$ before all of $S_{\text{right}}$.
To order the shifts within each of $S_{\text{left}}$ and $S_{\text{right}}$, repeat this process recursively.
If we can do the partitioning at each level of the recursion in a deterministic way, then we end up with a deterministic ordering of all possible shifts such that: for any $2^\ell$, if we break $\{0,\ldots,N-1\}$ into $2^{n-\ell}$ consecutive blocks of $2^\ell$ shifts, then the Fourier coefficients of each of those blocks of shifts have bounded 2-norm, and thus we can rely on \Cref{thm:small-fourier-implies-orn}.

In order to do the partitioning deterministically, we rely on discrepancy minimization theory.
In particular, a derandomized algorithm that achieves the Spencer's theorem bound for complex vector discrepancy minimization.
This result bounds the $\infty$-norm instead of the 2-norm, however we may use that to bound the 2-norm.

\subsection{Discrepancy Minimization for Complex Matrices} \label{sec:disc-minimization}

Levy, Ramdas, and Rothvoss prove the following:
\begin{thm}[Theorem $1$ of \cite{deterministic-discrepancy}] \label{thm:derand-disc}
  Let $v_1, \hdots, v_m \in \mathbb{R}^n$ unit vectors, let $x^{(0)} \in [-1, 1]^n$ be a starting point, and let $\lambda_1 \ge \hdots \ge \lambda_m \ge 0$ be parameters so that $\sum_{i = 1}^m \exp(-\lambda_i^2/16) \le n/32$.
  Then there is a deterministic algorithm that computes a vector $x \in [-1, 1]^n$ with $\langle v_i, x - x^{(0)}\rangle \le 8 \lambda_i$ for all $i \in [m]$ and $|\{i \, \mid \, x_i = \pm 1\}| \ge n/2$, in time $\mathcal{O}(\min\{n^4m, n^3 m\lambda_1^2\})$.
\end{thm}

From this, we get the following derandomization for vector discrepancy too.

\begin{thm}\label{thm:derand-matrix-disc}
  Let $A \in [-1, 1]^{m \times n}$. Then there exists a universal constant $K$ and a deterministic algorithm that finds an $x \in \{-1, +1\}^n$ such that
  \begin{align*}
    \norm{Ax}_\infty \le K \sqrt{n \log \frac{2m}{n}} .
  \end{align*}
\end{thm}

\begin{proof}
  Much of this follows from the proof of Theorem 2 from Lovett-Meka \cite{lovett-meka}.

  Let $A_k \in \mathbb{R}^n$ denote the $k$-th row of $A$. Let $v^k = A_k/\norm{A_k}_2$ and define $\alpha(x, y) = \sqrt{80\ln{2x/y}}$. Note that $m \cdot \exp \left(-\alpha(m, n)^2/16\right) < n/32$.

  The core idea is to run the algorithm from \Cref{thm:derand-disc} $\log_2 n$ many times, at each step fixing at least half the coordinates.

  Initialize the algorithm from \Cref{thm:derand-disc} with $x^{0} = \textbf{0}$. We obtain a vector $x^{1} \in [-1, 1]^n$ with $\langle v^k, x^1 \rangle \le 8 \alpha(m, n)$ for all $k \in [m]$ and $|\{j \, \mid \, x^1_j = \pm 1\}| \ge n/2$. Let $S = \{j \, \mid \, x_j^1 \ne \pm 1\}$. Then, we run the algorithm again starting at $x^{1}_S$, with the family of unit vectors $\{v^k/\norm{v^k}_2 \, \mid \, k \in S\}$ and $\lambda_i = \alpha(m, |S|)$ to find a vector $x_S^{2}$ with $\langle v^k_S, x_S^{2} - x_S^{1}  \rangle \le 8 \norm{v^k_S}_2 \alpha (m, |S|)$ for all $k \in S$. Repeat this iterative procedure at most $\log_2 n$ times to obtain a vector $x$ so that $|\{j \, \mid \, x_j = \pm 1\}| = n$. The discrepancy at the end then is
  \begin{align*}
    \norm{Ax}_\infty \le \sum_{i = 0}^{\log_2 n} \sqrt{n/2^i}\alpha(m, n/2^{i}) & = 8 \sqrt{80 n} \sum_{i = 0}^{\log_2 n} \sqrt{\frac{\ln 2^{i + 1} m/n}{2^i}} \le K\sqrt{n \ln 2m/n}
  \end{align*}
  for some universal constant, $K$, with the last inequality following by noting that the sum can be dominated by a convergent geometric series.

\end{proof}

\begin{cor}\label{cor:derand-complex-matrix-disc}
  Let $A \in \mathbb{C}^{m \times n}$ with $\norm{A_{j, k}}_1 \le 1$ for all $j, k$. Then there exists a universal constant $K$ and a deterministic algorithm that finds an $x \in \{-1, +1\}^n$ such that
  \begin{align} \label{eq:cor-dcmd}
    \norm{A x}_\infty \le K \sqrt{n \log \frac{4m}{n}} .
  \end{align}
  Furthermore, if $n$ is even, $x$ may be chosen to satisfy
  $x_1 + \cdots + x_n = 0$.
\end{cor}

\begin{proof}
	To create a matrix $D \in \mathbb{R}^{(2m+1) \times n}$, take each entry $a_{j,k}$ of $A$, split it into its real and imaginary components, and (effectively) replace the entry of $A$ with the 2-dimensional vector ${\Re(a_{j,k}) \choose \Im(a_{j,k})}$, doubling the height of each column of $A$. 
	Then, augment the matrix with an initial row of 1's, resulting in a matrix $D$ having $2m+1$ rows and $n$ columns. 
	Applying \Cref{thm:derand-matrix-disc} to $D$ we deduce that there is a vector $\tilde{x} \in \{ \pm 1 \}^n$ such that 
	\begin{equation}\label{eq:tildex}
		\norm{D \tilde{x}}_{\infty} \leq K \sqrt{n \log \frac{4m+2}{n}}  
	\end{equation} 
	If $n$ is odd we set $x = \tilde{x}$. If $n$ is even, the vector  $\tilde{x}$ may fail to satisfy the equation  $\tilde{x}_1 + \cdots + \tilde{x}_n = 0$. 
	Let $\alpha = |\tilde{x}_1 + \cdots + \tilde{x}_n|$.
	Since the first row of $D$ is made up of 1's, the first coordinate of the vector $D \tilde{x}$ is  equal to $\pm \alpha$. 
	By Inequality~\eqref{eq:tildex}  this implies $\alpha \leq K \sqrt{n \log \frac{4m+2}{n}}.$
	Modify $\tilde{x}$ into a vector $x$ that satisfies  $x_1 + \cdots + x_n = 0$ by reversing the signs of  $\alpha/2$ coordinates. 
	In other words, find a vector $y \in \{ -2, 0, 2 \}^n$ with exactly $\alpha/2$ non-zero  coordinates, such that $x = \tilde{x} + y$ belongs to $\{ \pm 1 \}^n$ and satisfies $x_1 + \cdots + x_n = 0$. We have $\norm{y}_1 = 2 \cdot (\alpha/2) = \alpha$, and
	\begin{equation}\label{eq:tildex.2}
		\norm{D x}_{\infty} \leq \norm{D \tilde{x}}_{\infty} + \norm{D y}_{\infty}
		\leq \norm{D \tilde{x}}_{\infty} + \norm{y}_1
		\leq 2 K \sqrt{n \log \frac{4m+2}{n}}.
	\end{equation}
	Inequality~\eqref{eq:cor-dcmd} above can be justified by redefining the universal constant $K$ in the corollary to be three times the constant $K$ that appears in this proof.
\end{proof}

\subsection{Subdividing Columns of the Fourier Matrix}

Assume that $N$ is a power of 2, $N=2^n$.
Let $\mathcal{W} \in \mathbb{C}^{N \times N}$ denote the Fourier transform matrix, defined by
$\mathcal{W}_{j,k} = e^{\frac{2\pi i}{N}jk}$ for all $j,k \in \set{0, \ldots, N - 1}$.
Let $\mathbf{1} \in \mathbb{C}^N$ denote the vector with all entries equal to one.
Given a set $S \subseteq \set{0, \ldots, N - 1}$, let $\mathcal{W}_S$ denote the submatrix obtained by selecting just the columns in $S$ from $\mathcal{W}$.
For ease of notation, when $S = \{j\}$ is a singleton, we just write $\mathcal{W}_j$.
We also let $W^\star \in \mathbb{C}^{N \times N}$ denote the Fourier transform matrix with first row entries set to 0.

We will repeatedly apply \Cref{cor:derand-complex-matrix-disc} to subdivide sets of even cardinality into two subsets of equal size, as depicted in \Cref{fig:binary-tree}. To subdivide a given set $S$ in the diagram, denote the cardinality of $S$ by $2j$ and apply
\Cref{cor:derand-complex-matrix-disc} using the matrix $A = \mathcal{W}_S$,
to obtain a vector $x \in \{ \pm 1 \}^{2j}$ satisfying $x_1 + \cdots + x_{2j} = 0$. Then define the two ``child sets'' of $S$ to be $S^- = \{ s \mid x_s = -1 \}$ and $S^+ = \{s \mid x_s = +1 \}$. The sets produced by this construction obey the following lemma.

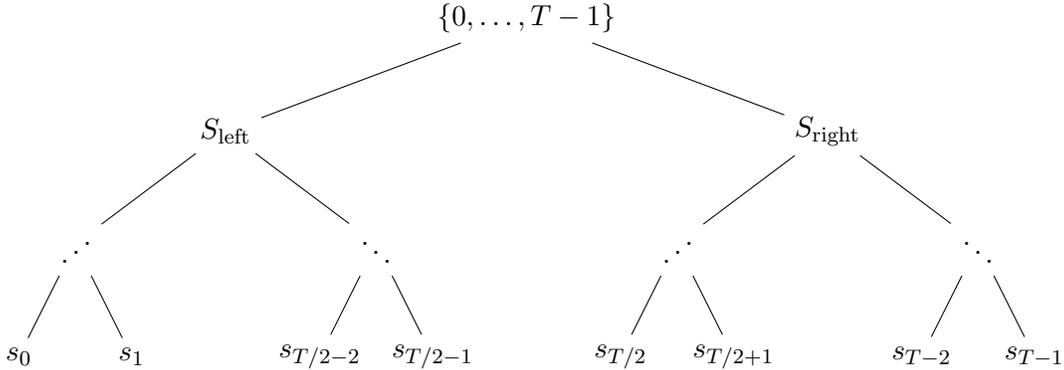
\begin{figure}[h]
  \begin{center}
    \begin{tikzpicture}[level distance=1.5cm,
        level 1/.style={sibling distance=8cm},
        level 2/.style={sibling distance=4cm},
        level 3/.style={sibling distance=1.5cm}]
      \node {$\set{0, \ldots, T - 1}$}
      child {node {$S_{\text{left}}$}
          child {node {$\iddots$}
              child{node {$s_0$}}
              child{node {$s_1$}}}
          child {node {$\ddots$}
              child {node {$s_{T/2 - 2}$}}
              child {node {$s_{T/2 - 1}$}}}
        }
      child {node {$S_{\text{right}}$}
          child {node {$\iddots$}
              child{node {$s_{T/2}$}}
              child{node {$s_{T/2 + 1}$}}}
          child {node {$\ddots$}
              child{node {$s_{T - 2}$}}
              child{node {$s_{T - 1}$}}}
        };
    \end{tikzpicture}
  \end{center}
  \caption{Recursive partitioning scheme for obtaining a
    total ordering of $S$}
  \label{fig:binary-tree}
\end{figure}
\begin{lem} \label{lem:det-sched-fourier-bound}
  Start with the total set $\{0,\ldots,N-1\}$, and assume that $N=2^{n}$ for some integer $n$.
  After using \Cref{cor:derand-complex-matrix-disc} recursively $\ell$ times, $\{0,\ldots,N-1\}$ is partitioned into $2^{\ell}$ blocks, each containing $2^{n-\ell}$ shifts.
  Call the total set of these blocks $\mathcal{A}_\ell$.
  Then for every $A\in \mathcal{A}_\ell$ for $\ell>0$,
  \begin{align*}
    \norm{\calW^\star_A \mathbf{1}}_{\infty} \leq C \cdot 2^{-\ell/2} \sqrt{N \ell}
  \end{align*}
  for a constant $C$.
\end{lem}

\begin{proof}
  We will show a bound of the form $\norm{\calW^\star_A \mathbf{1}}_{\infty} \leq M_\ell$ whenever $\ell>0$ and $A \in \mathcal{A}_{\ell}$, where the value of $M_{\ell}$ is given by a recursive formula. Then we will use the recurrrence to show the desired closed formula for $M_{\ell}$.

  \textbf{Claim:} When $0 \leq \ell \leq n$, then for every $A\in \mathcal{A}_\ell$, $\norm{\calW^\star_A \mathbf{1}} \leq M_{\ell}$ for $M_\ell$ with the recursive formula
  \begin{align*}
    M_\ell & = \begin{cases}
                 0                                                                                     & \mbox{if } \ell=0     \\
                 \frac{1}{2} M_{\ell - 1} + \frac{1}{2}K \cdot \sqrt{N \cdot 2^{-\ell + 1} (\ell + 1)} & \mbox{if } \ell > 0 .
               \end{cases}
  \end{align*}
  The claim is proven by induction on $\ell$. In the base case, $\mathcal{A}_0$ only contains the set $S=\{0,\ldots,N-1\}$, and therefore
  \[ \norm{\calW^\star_S \mathbf{1}}_\infty = \norm{\calW^\star \mathbf{1}}_\infty = 0 . \]
  
  Now consider the recursive case $\ell > 0$.
  By inductive hypothesis, assume for every $A\in \mathcal{A}_{\ell-1}$, $\norm{\calW^\star_A \mathbf{1}} \leq M_{\ell-1}$.
  Consider that in our next level, any set $B\in \mathcal{A}_\ell$ is created by first taking some set $A\in \mathcal{A}_{\ell-1}$ and partitioning it using a vector $x \in \{ \pm 1 \}^{|A|}$ into equal-sized pieces, $A_{\text{left}}$ and $A_{\text{right}}$, one of which equals $B$. The vector in
  $\{0,1\}^{|A|}$ identifying the elements of $B$ is equal to one of the
  vectors $(\bm{1} + x)/2$ or $(\bm{1}-x)/2$.
  By \Cref{cor:derand-complex-matrix-disc},
  \begin{align*}
    \norm{\calW^\star_{B} \mathbf{1}}_\infty 
     & \leq \frac{1}{2} (\norm{\calW_A^\star \mathbf{1}}_\infty + \norm{\calW_A^\star x}_\infty)         \\
     & \leq \frac{1}{2}\left( M_{\ell-1} + K\sqrt{|A| \log_2\left(4N/|A| \right)} \right)                \\
     & = \frac{1}{2}\left( M_{\ell-1} + K\sqrt{2^{n-\ell+1} \log_2\left(4N/2^{n-\ell+1} \right)} \right) \\
     & = \frac{1}{2}\left( M_{\ell-1} + K\sqrt{N\cdot 2^{-\ell+1} (\ell+1) } \right) .
  \end{align*}

  Now that we have proved the recursive formula of $M_\ell$, we use it to find a closed formula.
  \begin{align*}
    M_\ell & = \frac{1}{2} K \sum_{j=1}^{\ell-1} \left(\frac{1}{2}\right)^{\ell-1-j} \sqrt{2^{n-j}(j+1)} \\
           & = \frac{1}{2} K \sum_{j=1}^{\ell-1} \sqrt{ \frac{2^{n-j}(j+1)}{2^{2(\ell-1-j)}} }           \\
           & = \frac{1}{2^\ell}K \sum_{j=1}^{\ell-1} \sqrt{ 2^{n+j}(j+1) }                               \\
           & = \frac{1}{2^\ell}K \sqrt{N} \sum_{j=1}^{\ell-1} \sqrt{ 2^{j}(j+1) } .
  \end{align*}
  Finally, we simplify this closed formula to obtain the desired bound.
  \begin{align*}
    M_\ell & \leq K \frac{1}{2^{\ell}} \sqrt{N(\ell + 1)} \sum_{j = 1}^{\ell-1} 2^{j/2} \\
           & \le K 2^{-\ell} \sqrt{2N\ell} \cdot \frac{{2}^{\ell/2} - 1}{\sqrt{2} - 1}  \\
           & \le \frac{4 K}{\sqrt{2} - 1} \cdot 2^{-\ell/2} \sqrt{N \ell}.
  \end{align*}

\end{proof}

\subsection{Applying \Cref{thm:small-fourier-implies-orn}}

\begin{thm} [\textbf{\Cref{thm:main-thm}.3}] \label{thm:det-sched-works}
  Let $\eps>0$ be a constant, and assume $N = 2^n$ is a power of 2.
  Let $\bm{\sigma} = s_0,\ldots,s_{N-1}$ be the deterministic shift connection schedule created by repeatedly applying \Cref{cor:derand-complex-matrix-disc} $n$ times on $\{0,\ldots,N-1\}$ and the Fourier transform matrix $\calW$.
  Then $\bm{\sigma}$ can be efficiently computed.
  Additionally for any hop count $h$, there exists an oblivious routing protocol on $\bm{\sigma}$ that guarantees throughput $\frac{1}{2h}(1-\eps)$ and attains maximum latency no more than $\bigo_\eps\left(L^\star(h,N) \cdot\log N\right)$.
\end{thm}

\begin{proof}
  Note that for hop count $h=1$, we may use phase length $\Lambda_1=N$.
  In this case, $\hat{p}_{1,t}^\star = \calW^\star \mathbf{1}$, and thus has 2-norm equal to 0.

  Let $h>1$ be some hop count,
  and let $\Lambda_h = 2^g$ be the smallest power of 2 for which
  \[ \Lambda_h = 2^g \geq \frac{C^2 (4N)^{1/h}}{\eps^{2/h}}\cdot\log_2 N , \]
  where $C$ is the universal constant from \Cref{cor:derand-complex-matrix-disc}. 
  Note that $\Lambda_h$ divides $N=2^n$.
  Then applying \Cref{thm:small-fourier-implies-orn}, we only need to check that $\norm{\hat{p}_{h,t}^\star}_2\leq \frac{\eps}{2}$ and $\norm{\hat{q}_{h,t}^\star}_2 \leq \frac{\eps}{2}$ for all starting timesteps $t$ that are multiples of $\Lambda_h$.

  Recall that
  \[ \hat{p}_{h,t}^\star = \bigodot_{j=1}^{h} \hat{p}_{h,t,j}^\star \]
  where
  \begin{align*}
    \hat{p}_{h,t,j}^\star = \frac{1}{\Lambda_h}
    \begin{pmatrix}
      0                                                                              \\
      \vspace{2mm} 
      \sum_{k=t+(j-1)\Lambda_h}^{t+j\Lambda_h-1} \mbox{ } e^{\frac{2\pi i}{N} s_{k}} \\
      \sum_{k=t+(j-1)\Lambda_h}^{t+j\Lambda_h-1} \mbox{ } e^{\frac{4\pi i}{N} s_{k}} \\
      \vdots                                                                         \\
      \sum_{k=t+(j-1)\Lambda_h}^{t+j\Lambda_h-1} \mbox{ } e^{\frac{(N-1)\pi i}{N} s_{k}}
    \end{pmatrix} ,
  \end{align*}
  and that $\hat{q}_{h,t}^\star$ is similarly defined.
  We will bound their 2-norms by first bounding $\norm{\hat{p}_{h,t,j}^\star}_\infty$.

  Notice that because $t$ is a multiple of $\Lambda_h = 2^g$, then $\hat{p}_{h,t,j}^\star = \frac{1}{\Lambda_h} \calW_A^\star$ for some set $A \in \mathcal{A}_{n-g}$.
  Then we may use \Cref{lem:det-sched-fourier-bound} to bound the $\infty$-norm.
  \begin{align*}
    \norm{\hat{p}_{h,t,j}^\star}_\infty & \leq \frac{1}{\Lambda_h} \cdot C \cdot 2^{-(n-g)/2} \sqrt{N (n-g)} \\
                                        & =  C \cdot
    \left( \frac{1}{\Lambda_h} \cdot 2^{g/2} \right) \cdot
    \left( 2^{-n/2} \cdot \sqrt{N} \right) \cdot \sqrt{n-g}
    \\
                                        & = C \cdot \frac{1}{\sqrt{\Lambda_h}} \cdot \sqrt{n-g}              \\
                                        & \leq \sqrt{\frac{C^2 \log_2 N}{\Lambda_h}}                         \\
                                        & \leq \frac{\left(\frac{\eps}{2}\right)^{1/h}}{N^{1/2h}}
  \end{align*}
  where the last line follows from our assumption that
  $\Lambda_h \geq  \frac{C^2 (4N)^{1/h}}{\eps^{2/h}}\cdot\log_2 N $.
  This implies a bound on the $2h$-norm.
  \begin{align*}
    \norm{\hat{p}_{h,t,j}^\star}_{2h} & \leq \left( \left( \frac{\left(\frac{\eps}{2}\right)^{1/h}}{N^{1/2h}} \right)^{2h} \cdot N \right)^{1/2h} \\
                                      & = \left(\frac{\eps}{2}\right)^{1/h} .
  \end{align*}
  We can then apply H\"{o}lder's Inequality.
  \[
    \norm{\hat{p}_{h,t}^\star}_{2} = \norm{\bigodot_{j=1}^{h} \hat{p}_{h,t,j}^\star}_{2} \leq \prod_{j=1}^{h} \norm{\hat{p}_{h,t,j}^\star}_{2h} \leq \left(\left(\frac{\eps}{2}\right)^{1/h}\right)^h = \frac{\eps}{2} .
  \]
  A similar bound can be derived for $\norm{\hat{q}_{h,t}^\star}_2$.
  Therefore, applying \Cref{thm:small-fourier-implies-orn} gives our desired guaranteed throughput bound $\frac{1}{2h}(1-\eps)$.

  Note that the oblivious routing protocol VLB with Leakage attains a maximum latency of $(2h+1)\Lambda_h$.
  Because $\Lambda_h \leq \frac{2 C^2 (4N)^{1/h}}{\eps^{2/h}}\cdot\log_2 N$, then $(2h+1)\Lambda_h \leq \bigo(hN^{1/h}\cdot\log N) = \bigo\left(L^\star(h,N) \cdot \log N\right)$.

\end{proof}

\section{Conclusion}\label{sec:conclusion}

In this paper, we proved the existence of universal connection schedules for oblivious reconfigurable networks (ORNs).
We described three different constructions, two random and one deterministic which each achieved slightly different near-optimal tradeoffs between guaranteed throughput and maximum latency.
Below, we discuss some extended results and possible future directions.

\subsection{Multiple Traffic Classes, Simultaneously}\label{sec:multi-traffic}

One motivation of a universal connection schedule is to satisfy heterogeneous workloads containing a mix of low- and high-latency traffic requests.
A natural question, then, is how much throughput for each traffic class can be fulfilled simultaneously on the ORN designs we construct in this paper?
We limit discussion in this section to the uniformly randomized design of \Cref{sec:rand-shift-sched}, though the same ideas apply to our other designs with similar results.

Consider a workload made up of connection requests each with a specified source, destination, starting timestep, and maximum latency. 
The workload can be partitioned into \emph{traffic classes} indexed by positive integers, with traffic class $h$ representing the subset of connection requests whose maximum latency is between $2 (h+1) \Lambda_{h+1}$ and $2 h \Lambda_h$, where $\Lambda_h = \bigotilde_\eps (N^{1/h})$ is as defined in \Cref{sec:rand-shift-sched}. 
For workloads whose connection requests all belong to a single traffic class, $h$, the randomized ORN design of \Cref{sec:rand-shift-sched} guarantees throughput $\frac{1}{2h}(1-\eps)$.
When routing multiple traffic classes simultaneously, let us start by assuming that each traffic class $h$ has its own desired throughput rate $r_h$, fixed across all time.
\begin{prop} \label{prop:mult-traffic-fixed-time}
	Consider the randomized ORN design of \Cref{sec:rand-shift-sched}, which includes a single connection schedule and oblivious routing protocols $R_h$ for each traffic class $h$.
	Let $r_h \in [0,\frac{1}{2h}(1-\eps)]$ be the desired throughput rate of traffic class $h$, and suppose that
	\[ \sum_h \frac{2h}{1-\eps}r_h \leq 1 . \]
	Then each traffic class $h$ may continuously send at rate $r_h$ simultaneously without overloading the network.
	That is, for every set of demand functions $D_h$ satisfying
	\[ \forall a \in [N], \, t \in \mathbb{Z} \quad
		\sum_{b \in [N]} D(a,b,t) \leq r_h \quad \mbox{and} \quad
		\sum_{b \in [N]} D(b,a,t) \leq r_h ,
	\]
	the sum of induced flows $\sum_h R_h^{D_h}$ is feasible.
\end{prop}

Requiring that each traffic class $h$ maintain the same desired throughput rate across all time is quite restrictive.
We may relax that assumption by introducing a desired throughput rate $r_{h,t}$ for traffic class $h$ at time $t$.

\begin{prop} \label{prop:mult-traffic-varying-time}
	Let $r_{h,t}$ be the desired throughput rate of traffic class $h$ at timestep $t$, and suppose that
	\[ \forall t^* \quad \sum_h \sum_{t\in [t^*-2h\Lambda_h+1,t^*]} \frac{r_{h,t}}{(1-\eps)\Lambda_h} \leq 1 . \]
	Then each traffic class may send at rate $r_{h,t}$ simultaneously without overloading the network.
	That is, for every set of demand functions $D_h$ satisfying
	\[ \forall a \in [N], \, t \in \mathbb{Z} \quad
		\sum_{b \in [N]} D(a,b,t) \leq r_{h,t} \quad \mbox{and} \quad
		\sum_{b \in [N]} D(b,a,t) \leq r_{h,t} ,
	\]
	the sum of induced flows $\sum_h R_h^{D_h}$ is feasible.
\end{prop}
The proofs of Propositions 2 and 3 are simple counting arguments and are thus omitted.

To understand why this more complex relationship exists when traffic classes request varying throughput rates over time, consider the simple case when there are only two traffic classes, class 1 and class 2.
Up until timestep $t^*$, class 1 requests the largest through it can, $\frac{1}{2}(1-\eps)$ while class 2 requests $0$. Starting at timestep $t^*$, class 1 requests 0, while class 2 would like to request the largest throughput it can, $\frac{1}{4}(1-\eps)$.
However, class 1 still has traffic in the network that has not yet been fully delivered to its final destination.
This means that the full network bandwidth is not immediately available for class 2 to use.
It must wait for the class 1 traffic to clear through the network before it can send at the full throughput rate $\frac{1}{4}(1-\eps)$.
Instead, at timestep $t^*$, class 2 may start sending at throughput rate $\frac{(1-\eps)\Lambda_2}{2\Lambda_1}$ for this and all future timesteps.
At timestep $t^* + 4\Lambda_2$, the first class 2 traffic begins to reach its destination, and thus class 2 may increase its throughput rate to $\frac{(1-\eps)\Lambda_2}{\Lambda_1}$ for this and all future timesteps.
This process of increasing class 2's throughput rate every $4\Lambda_2$ timesteps continues until timestep $t^* + 2\Lambda_1$, when all class 1 traffic has been delivered to its destination and class 2 may finally send at its maximum throughput rate $\frac{1}{4}(1-\eps)$.

\subsection{Extending the VLB with Leakage Framework to Other Models}\label{sec:leak-exten}

\Cref{sec:vlb-w-leakage} formalizes a general framework for applying VLB-like routing when flow is not distributed uniformly among all
intermediate nodes. We believe this framework is useful in many network settings, not solely for shift connection schedules on ORNs.

The routing distributions from VLB with Leakage take exactly $h$ random hops, with each hop taken uniformly at random from consecutive non-overlapping phases of $\Lambda_h$ timesteps.
At the end of these $h$ hops, we hope to be at an approximately uniformly random node, no matter which source node (or, destination node when considering backwards routing distributions) we started at.
Throughout this paper, we limited analysis of VLB with Leakage to shift connection schedules only.
Because we dealt with shift connection schedules, we only needed to check that the routing distributions were close to uniformly random for a single source node (or, single destination node).
We could then apply the spatial symmetry of shift connection schedules to show this implied close to uniformly random routing distributions for every source (or destination) node.
Additionally, one can check this criterion using generating polynomials that are multiplied together.
This was the main motivation for limiting to shift connection schedules.
The generating polynomial approach does not work when we relax to non-shift connection schedules.

A more general approach that would apply to non-shift connection schedules would be to use Markov transition matrices.
For each phase $k$, one can compute the Markov transition matrix $M_k$ of the phase.
Let $M_k[i,j]$ denote the probability that a message starting at node $i$ at the beginning of phase $k$, and taking a uniformly random hop during the phase, ends at node $j$. Then $M_k$ is an average of $\Lambda_h$ permutation matrices, each representing the permutation of connections for each timestep in the phase.
The matrix product $\mathcal{M} = M_h\hdots M_2 M_1$ is therefore the probability matrix representing the routing distributions from VLB with Leakage, using $h$ hops with phase length $\Lambda_h$.
In fact, each row $i$ of $\mathcal{M}$ is exactly the forward routing distribution when starting at source node $i$. (Or, one could compute a similar matrix product for the backwards routing distributions.)
To show that each routing distribution is $\frac{\eps}{2}$-close to uniform, one must bound the maximum row 1-norm (also known as the $\ell_\infty$ operator norm) of $\mathcal{M} - U$, where $U$ is the uniform $N\times N$ matrix with all entries equal to $\frac{1}{N}$.
Thus, to test whether any fixed non-shift connection schedule is universal when using the routing distributions from VLB with Leakage, one only needs to compute these matrix products for each $h \in \{1,\hdots,\log_2(N)\}$ and verify that when $U$ is subtracted, they have low $\ell_\infty$ operator norm.
This can be done in polynomial time in the size of the network.

Proving that a random non-shift connection schedule is universal is more difficult. 
One would need a way to bound the probability that the product of $h$ matrices, each an average of $\Lambda_h$ random permutation matrices sampled from a specified distribution, is far from $U$ in $\ell_\infty$ operator norm. 
For shift connection schedules, all of the matrices in question are circulant matrices, so one can conjugate by the Fourier transform matrix to simultaneously diagonalize them. 
This change of basis (which amounts to a reinterpretation of the Fourier-analytic approach we adopted in \Cref{sec:small-fourier-implies-orn}) was a key step enabling our analysis of the constructions in this paper. 
For connection schedules whose constituent permutations belong to some other subgroup $\Gamma$ of the permutation group $S_N$, the action of $\Gamma$ on $\reals^N$ splits into a direct sum of irreducible representations, and the analogous change of basis conjugates each permutation matrix represented by $\Gamma$ to a block-diagonal matrix with one block for each irreducible representation in this decomposition. 
Can this representation-theoretic perspective, which generalizes our Fourier-analytic approach, lead to useful bounds on the $\ell_\infty$ operator norms of matrices $\mathcal{M} - U$ for schedules based on non-abelian subgroups $\Gamma \subseteq S_N$? We leave this question to future work.

\noindent\textbf{Static networks} 

One could even apply the VLB with Leakage framework in the context of static networks.
In static networks the hop parameter $h$ translates cleanly, however the analogous parameter to phase length $\Lambda_h$ is the degree of the static network $d$.
Instead of taking $h$ hops uniformly at random, each from consecutive non-overlapping timesteps, the analogous procedure in a static network would be an $h$-hop random walk through the network.
At the end of this $h$-hop random walk, in order to use the VLB with Leakage framework, we hope to be at an approximately uniformly random node.
Similarly to non-shift connection schedules on ORNs, we can use Markov transition matrices to bound how close we are to uniformly random after an $h$-hop random walk starting from any node.

In fact, the resulting forward and backwards routing distributions from applying the VLB with Leakage framework to static networks are exactly the Markov chain that represents taking a random walk on a static network.
If we take an $h$-hop random walk in the network, how well mixed (up to what parameter $\eps$) does the resulting distribution end up being?
This is the same as the mixing time of the network, although quantifiers have been flipped compared to their usual ordering.
From this, we can state the following informal Corollary.

\begin{cor}\label{thm:mixing-implies-obliv-rout}
	Given an unweighted graph $G$ of degree $d$ with $\frac{\eps}{2}$-mixing time $h$, there exists an oblivious routing protocol which uses paths of length $2h$ and guarantees maximum congestion no more than $2h\left(\frac{1}{1-\eps}\right)$ so long as traffic demands are ``reasonable,'' that is nodes do not request to send or receive more demand than their degree, $d$.
\end{cor}

Mixing times and their connections to oblivious routing protocols have been studied before.
In fact, previous work~\cite{obliv-rout-rand-walks-17} noted a similar conclusion, that one can build oblivious routing protocols with provable guarantees on congestion by concatenating paths carefully sampled from a random walk process.
However, they concatenate paths differently, and they express their results as the competitive ratio and in terms of the spectral gap of the graph.

\endgroup

\bibliographystyle{alpha}
\bibliography{biblio}{}

\end{document}